\documentclass[10pt]{article}

\usepackage[margin=2cm]{geometry}

\usepackage{hyperref}
\usepackage{appendix}

\usepackage{times,graphicx,epic,eepic,epsfig,amsmath,latexsym,amssymb,verbatim,revsymb}

\newcommand{\extra}[1]{}

\renewcommand{\comment}[1]{}

\usepackage{amsmath,amssymb,amsfonts}
\usepackage{amsthm}

\newtheorem{theorem}{Theorem}[section]

\newtheorem{corollary}[theorem]{Corollary}

\newtheorem{lemma}[theorem]{Lemma}

\newtheorem{proposition}[theorem]{Proposition}

\theoremstyle{remark}

\def\squareforqed{\hbox{\rlap{$\sqcap$}$\sqcup$}}
\def\qed{\ifmmode\squareforqed\else{\unskip\nobreak\hfil
\penalty50\hskip1em\null\nobreak\hfil\squareforqed
\parfillskip=0pt\finalhyphendemerits=0\endgraf}\fi}
\def\endenv{\ifmmode\;\else{\unskip\nobreak\hfil
\penalty50\hskip1em\null\nobreak\hfil\;
\parfillskip=0pt\finalhyphendemerits=0\endgraf}\fi}

\renewenvironment{proof}{\noindent \textbf{{Proof~} }}{\qed\medskip}
\newenvironment{proof+}[1]{\noindent \textbf{{Proof #1~} }}{\qed\medskip}
\newenvironment{remark}{\noindent \textit{{Remark.~}}}{\qed}

\mathchardef\ordinarycolon\mathcode`\:
\mathcode`\:=\string"8000
\def\vcentcolon{\mathrel{\mathop\ordinarycolon}}
\begingroup \catcode`\:=\active
  \lowercase{\endgroup
  \let :\vcentcolon
  }

\newcommand{\nc}{\newcommand}
\nc{\rnc}{\renewcommand}
\nc{\beq}{\begin{equation}}
\nc{\eeq}{{\end{equation}}}
\nc{\beqa}{\begin{eqnarray}}
\nc{\eeqa}{\end{eqnarray}}
\nc{\lbar}[1]{\overline{#1}}
\nc{\bra}[1]{\langle#1|}
\nc{\ket}[1]{|#1\rangle}
\nc{\ketbra}[2]{|#1\rangle\!\langle#2|}
\nc{\braket}[2]{\langle#1|#2\rangle}
\nc{\proj}[1]{| #1\rangle\!\langle #1 |}
\nc{\avg}[1]{\langle#1\rangle}
\nc{\smfrac}[2]{\mbox{$\frac{#1}{#2}$}}
\nc{\tr}{\operatorname{tr}}
\nc{\tracedist}[1]{\Delta_{}\!\left( #1 \right)}
\nc{\fid}[1]{F\!\left( #1 \right)}

\newcommand{\sket}[1]{|{#1}\rangle\rangle}
\newcommand{\sbra}[1]{\langle\langle {#1}|}

\nc{\ox}{\otimes}
\nc{\dg}{\dagger}
\nc{\dn}{\downarrow}
\nc{\cA}{{\cal A}}
\nc{\cB}{{\cal B}}
\nc{\cC}{{\cal C}}
\nc{\cD}{{\cal D}}
\nc{\cE}{{\mathcal E}}
\nc{\cF}{{\cal F}}
\nc{\cG}{{\cal G}}
\nc{\cH}{{\cal H}}
\nc{\cI}{{\cal I}}
\nc{\cJ}{{\cal J}}
\nc{\cK}{{\cal K}}
\nc{\cL}{{\cal L}}
\nc{\cM}{{\cal M}}
\nc{\cN}{{\cal N}}
\nc{\cO}{{\cal O}}
\nc{\cP}{{\cal P}}
\nc{\cR}{{\cal R}}
\nc{\cS}{{\cal S}}
\nc{\cT}{{\cal T}}
\nc{\cU}{{\cal U}}
\nc{\cV}{{\cal V}}
\nc{\cX}{{\cal X}}
\nc{\cZ}{{\cal Z}}

\nc{\entI}{{\bf I}}
\nc{\entIarrow}{{\bf I}^{\leftarrow}}
\nc{\entH}{{\bf H}}
\nc{\entS}{{\bf S}}
\nc{\entHmin}{\mathbf{H}_{\min}}
\nc{\aentHmin}{\hat{\mathbf{H}}_{\min}}

\nc{\supp}{\textrm{supp}}
 
\nc{\entF}{{\bf E}_f}

\nc{\isom}{\simeq}

\nc{\rank}{\operatorname{rank}}
\nc{\rar}{\rightarrow}
\nc{\lrar}{\longrightarrow}
\nc{\polylog}{\operatorname{polylog}}
\nc{\poly}{\operatorname{poly}}
\nc{\1}{{\openone}}

\nc{\weight}{\textbf{w}}
\nc{\hamdist}{d_{H}}

\def\e{\epsilon}

\nc{\Sp}{{{\mathbb S}}}
\nc{\RR}{{{\mathbb R}}}
\nc{\CC}{{{\mathbb C}}}
\nc{\FF}{{{\mathbb F}}}
\nc{\NN}{{{\mathbb N}}}
\nc{\ZZ}{{{\mathbb Z}}}
\nc{\PP}{{{\mathbb P}}}
\nc{\QQ}{{{\mathbb Q}}}
\nc{\UU}{{{\mathbb U}}}
\nc{\OO}{{{\mathbb O}}}
\nc{\EE}{{{\mathbb E}}}
\nc{\id}{{\operatorname{id}}}

\nc{\qubitchannel}{\id_2}
\nc{\bitchannel}{\overline{\id}_2}

\nc{\be}{\begin{equation}}
\nc{\ee}{{\end{equation}}}
\nc{\bea}{\begin{eqnarray}}
\nc{\eea}{\end{eqnarray}}
\nc{\<}{\langle}
\rnc{\>}{\rangle}
\nc{\Hom}[2]{\mbox{Hom}(\CC^{#1},\CC^{#2})}
\nc{\rU}{\mbox{U}}

\nc{\ob}[1]{#1}

\newcommand{\eqdef}	{\stackrel{\textrm{def}}{=}}

\newcommand{\ex}[1]	{\mathbf{E}\left\{ #1 \right\}}
\newcommand{\exc}[2]	{\underset{#1}{\mathbf{E}}\left\{ #2 \right\}}
\newcommand{\pr}[1]	{\mathbf{P}\left\{ #1 \right\}}

\newcommand{\event}[1]	{\left[ #1 \right]}
\newcommand{\eventfont}[1]	{\textsf{#1}}

\newcommand{\bin}	{\textrm{Bin}}

\renewcommand{\exp}[1]	{\operatorname{exp}\left( #1 \right)}

\nc{\unif}{\textrm{unif}}

\nc{\circuit}{\textrm{circ}}
\nc{\haar}{\textrm{haar}}

\begin{document}

\title{Scrambling speed of random quantum circuits}

\author{Winton Brown\thanks{D\'{e}partement de Physique, Universit\'{e} de Sherbrooke} \and Omar Fawzi\thanks{Institute for Theoretical Physics, ETH Z\"{u}rich}}

\date{\today}

\maketitle

\begin{abstract}
Random transformations are typically good at ``scrambling'' information.   Specifically, in the quantum setting, scrambling usually refers to the process of mapping most initial pure product states under a unitary transformation to states which are macroscopically entangled, in the sense of being close to completely mixed on most subsystems containing a fraction $fn$ of all $n$ particles for some constant $f$.  While the term scrambling is used in the context of the black hole information paradox, scrambling is related to problems involving decoupling in general, and to the question of how large isolated many-body systems reach local thermal equilibrium under their own unitary dynamics.

Here, we study the \emph{speed} at which various notions of scrambling/decoupling occur in a simplified but natural model of random two-particle interactions: random quantum circuits.  For a circuit representing the dynamics generated by a local Hamiltonian, the depth of the circuit corresponds to time. Thus, we consider the depth of these circuits and we are typically interested in what can be done in a depth that is sublinear or even logarithmic in the size of the system. We resolve an outstanding conjecture raised in the context of the black hole information paradox with respect to the depth at which a typical quantum circuit generates an entanglement assisted encoding against the erasure channel. In addition, we prove that typical quantum circuits of $\poly(\log n)$ depth satisfy a stronger notion of scrambling and can be used to encode $\alpha n$ qubits into $n$ qubits so that up to $\beta n$ errors can be corrected, for some constants $\alpha, \beta > 0$.
% a minimum distance of $fn$ for a constant fraction $f$.
\end{abstract}

\section{Introduction}

Random quantum circuits of polynomial size are meant to be efficient implementations that inherent many useful properties of ``uniformly'' chosen unitary transformations which are typically very inefficient. A lot of work was done in analyzing convergence properties of the distribution defined by random quantum circuits to the Haar measure on the full unitary group acting on $n$ qubits \cite{RM, ELL05, ODP07, Znidaric2, O09, HL09, Low10, BVPRL, BHH12}. Here, instead of trying to study the convergence of these circuits to some limit, we study the information-theoretic property of interest directly. This property can be intuitively pictured as ``scrambling'' or spreading some structured or localized information over the global system of $n$ qubits. The term scrambling is used in the context of the black hole information paradox \cite{HP07, SS08, LSHOH11}, but such a property can also be understood in terms of decoupling, a central notion in the study of quantum communication \cite{HOW05, HOW06, HHYW08, Dup09, ADHW09, DBWR10}. On a more technical level, a typical property of a scrambler can be seen when we decompose the input and output states in the Pauli basis (which can be seen as a Fourier basis): a scrambler tends to reduce the mass of the low-weight Pauli operators. In fact, all of our arguments prove a statement of that form.
%, and to the question of how large isolated many-body systems reach local thermal equilibrium under their own unitary dynamics.

\subsection{Strong scrambling, quantum error correction and decoupling}
An important example of a scrambler is an encoding circuit for a quantum error correcting code.  In particular a $k$-qubit, distance $d$, non-degenerate error correcting code maps all initial states localized on $k$ qubits to states which are completely mixed on all subsystems of size less than $d$, which can be considered a strong form of scrambling when the distance is a constant fraction of $n$. Another way of defining a good quantum error correcting code is that it decouples a purification of the encoded qubits from any subsystems of size smaller than $d$. Proving coding theorems by proving a decoupling statement has been quite successful culminating in a very general decoupling theorem \cite{HOW05, HOW06, HHYW08, Dup09, ADHW09, DBWR10}. Our objective can also be seen as trying to determine how fast decoupling occurs.

We prove the following results:
\begin{itemize}
\item We give a random quantum circuit model of depth $O(\log^3 n)$ that satisfies a strong notion of scrambling.  That is, for any initial state, on average over the circuit, all subsystems of size at most $fn$ are close to completely mixed.  
\item This result can also be considered as giving decoupling unitaries that are more efficient than standard (approximate) two-designs. Relying on the fact that random quantum circuits are approximate two-designs \cite{HL09}, it was shown by \cite{SDTR11} that random circuits of size $O(n^2)$ are decouplers in a quite general setting. Here we prove that in some particular cases, we can obtain much faster decoupling with circuits of depth $O(\log^3 n)$.
\item As another application, we prove the existence of stabilizer codes with an encoding circuit of depth $O(\log^3 n)$ that have a constant encoding rate and a minimum distance that grows linearly with $n$.  
%This result is related to some recent work on understanding the computational complexity of the encoding of good codes, in particular the depth \cite{GHKPV12} of such encoding circuits.
\end{itemize}

It would be interesting to prove that scrambling occurs in depth $O(\log n)$ instead. Our second set of results proves a weaker notion of scrambling in depth $O(\log n)$. This notion of scrambling is particularly relevant in the study of the black hole information paradox question.
%For this notion, we prove scrambling with better parameters.

\subsection{Black holes and the fast scrambling conjecture}
It was noted in \cite{DonPage}, that by collecting the Hawking radiation from a black hole, an arbitrary message dropped into the black hole could be recovered after half the black hole had evaporated if the dynamics of the black hole could be approximated as a random unitary transformation.  This approach was tightened significantly in \cite{HP07}, where it was shown that at any time after the black hole has evaporated past its half way point, an $m$-qubit quantum state that was dropped into the black hole could be recovered with high fidelity  from an amount of  Hawking radiation containing slightly more than $m$ qubits of quantum information, as long as the dynamics of the black hole approximates a unitary two-design sufficiently.   A random quantum circuit model analyzed in \cite{DCEL09} was invoked which could, for the purposes of recovering an initial state of constant size, scramble the degrees of freedom by a local circuit of depth $O(\log n)$.  This random quantum circuit model, though highly contrived, could be performed by two-qubit gates between nearest neighbours on a $2$-dimensional lattice, in a depth of $O(\sqrt{n}\log n)$.  
This amount of time is just enough to avoid a violation of the quantum no cloning principle assuming complementarity at the event horizon.  This motivated interest in the scrambling properties of more natural models of random quantum circuits that may better represent a naturally arising Hamiltonian. It was conjectured in \cite{SS08} that this was possible in time $O(n^{1/d})$ and $O(\log n)$ for a local Hamiltonian in $d$-dimensions and infinite dimensions respectively, with $k$-body interactions. Since the signaling bound precludes faster scrambling, such unitary transformations are referred to as ``fast scramblers".  
\iffalse
 In  \cite{LSHOH11}, significant progress was made towards proving this conjecture by considering an ensemble of Hamiltonians with two-body interactions over arbitrary distances, as well as a Lieb-Robinson bound on the speed of signaling for such infinite dimensional two-body Hamiltonians, which put an upper bound on the scrambling time.  However, due to a technical consideration only a ratio of scrambling times of subsystems of different sizes could be obtained, and in \cite{RQCL} such random quantum circuits can be shown to have diverging scrambling times. \fi
\begin{itemize}
\item Here we resolve the fast scrambling conjecture for random quantum circuits in the case of constant message size. We show that typical random quantum circuits on $d$-dimensional lattices and the complete graph, of depth $O(n^{1/d} \log^2 n)$ and $O(\log n)$ respectively,  scramble a message of constant size $m$ such that it may be recovered with high fidelity using only $m+c$ randomly selected qubits, for some constant $c$.  Since a straightforward lower bound of  $\Omega(n^{1/d})$ and $\Omega(\log n)$ can be shown, our results are nearly optimal.  
\end{itemize}

\subsection{Proof technique}
The first step of the proof is to relate the property of interest, which is most naturally stated in terms of the trace-norm, to the two norm, whose behavior under under the random quantum circuit can be completely described by its second-order moment operator. For the random quantum circuits we consider, this moment operator, when evaluated in the Pauli basis, can be seen as a Markov chain on the set of Pauli basis elements (also called Pauli strings). This means that the properties of interest can be seen as properties of this Markov chain.

Most previous studies of random quantum circuits bounded the convergence using the spectral gap of the moment operators. However, as we show, the spectral gap only weakly depends on the underlying interaction graph of the circuit.  This means that any result that uses the spectral gap of the second moment operator will give bounds on the scrambling time that would also apply to circuits where the gates are applied between neighbouring cells on a one dimensional line. In particular, as the diameter of the interaction graph plays a crucial role in determining the scrambling speed, it is necessary in our proofs to go beyond placing bounds on the spectral gap and to make use directly of the Markov chain (or a Markov chain obtained from lumping certain states), which heavily depends on the interaction graph.

% that are only tight  in the case where the gates are applied between neighbouring cells on a one dimensional line. In particular, as the diameter of the interaction graph determines the scrambling speed, it is necessary in our proofs to go beyond placing bounds on the spectral gap and to make use directly of the Markov chain.

\section{Preliminaries}

\subsection{Generalities}
The state of a pure quantum system is represented by a unit vector in a Hilbert space. Quantum systems  are denoted $A, B, C\dots$ and are identified with their corresponding Hilbert spaces. The Hilbert spaces we consider here will be mostly $n$-qubits spaces of the form $(\CC^2)^{\otimes n}$. To describe a distribution $\{p_1, \dots, p_r\}$ over quantum states $\{\ket{\psi_1}, \dots, \ket{\psi_r}\}$ (also called a mixed state), we use a density operator $\rho = \sum_{i=1}^r p_i \proj{\psi_i}$, where $\proj{\psi}$ refers to the projector on the line defined by $\ket{\psi}$. A density operator is a Hermitian positive semidefinite operator with unit trace. The density operator associated with a pure state is abbreviated by omitting the ket and bra $\psi \eqdef \proj{\psi}$. $\cS(A)$ is the set of density operators acting on $A$. The Hilbert space on which a density operator $\rho \in \cS(A)$ acts is sometimes denoted by a subscript, as in $\rho_A$. This notation is also used for pure states $\ket{\psi}_A \in A$.

In order to describe the joint state of a system $AB$, we use the tensor product Hilbert space $A \otimes B$, which is sometimes simply denoted $AB$.  If $\rho_{AB}$ describes the joint state on $AB$, the state on the system $A$ is described by the partial trace $\rho_A \eqdef \tr_B \rho_{AB}$. If $U$ is a unitary acting on $A$, and $\ket{\psi}$ a state in $A \ox B$, we sometimes use $U \ket{\psi}$ to denote the state $(U \ox \1_B) \ket{\psi}$, where the symbol $\1_B$ is reserved for the identity map on $B$. For an introduction to quantum information, we refer the reader to \cite{NC00}.

Throughout the paper, we use the Pauli basis, which is an orthogonal basis for $2 \times 2$ matrices:
\[
\sigma_0 = \1 \qquad 
\sigma_1 = \left(
\begin{array}{cc}
0 & 1 \\
1 & 0 \\
\end{array}
\right)
\qquad
\sigma_2 = \left(
\begin{array}{cc}
0 & -i \\
i & 0 \\
\end{array}
\right)
\qquad
\sigma_3 = \left(
\begin{array}{cc}
1 & 0 \\
0 & -1 \\
\end{array}
\right).
\]
For a string $\nu \in \{0,1,2,3\}^n$, we define $\sigma_{\nu} = \sigma_{\nu_1} \otimes \cdots \otimes \sigma_{\nu_n}$. The support $\supp(\nu)$ of $\nu$ is simply the subset $\{i \in [n] : \nu_i \neq 0\}$ and the weight $w(\nu) = |\supp(\nu)|$.

We now introduce some various notation. The notation $\poly(n)$ refers to a term that could be chosen to be any polynomial and the power of the polynomial can be made larger by appropriately choosing the related constants. As we are going to deal with binomial coefficients, the binary entropy function $h(x) = - x\log x - (1-x) \log(1-x)$ is going to be used. We also use the shorthand $[n] \eqdef \{1, \dots, n\}$.

%Standard quantum information stuff. Have some decomposition of into qubits. Notation for partial trace.
%Introduce Pauli basis.

\subsection{Random quantum circuits}
\label{sec:prelim-rqc}
We consider two related models for random quantum circuits. In a sequential random quantum circuit a random two-qubit gate is applied to a randomly chosen pair of qubits in each time step. For a general interaction graph, instead of choosing a pair at random from all the possible pairs, we choose a random edge in the graph. Here the random two-qubit gate is going to be a random Clifford gate or a gate uniformly chosen from the Haar measure on the unitary group acting on two qubits. In fact, as the second-order  moment operator is the same for these two models, our results apply equally well to them. However, the result of Theorem \ref{thm:good-codes} proving the existence of stabilizer codes with efficient encoding makes explicit use of the model with random Clifford gate.
%
% Should have a definition of what Clifford gate is but maybe in a next version.
%
%
Since we are interested in the speed at which scrambling occurs rather than the gate complexity, we ask %given a random sequence of two-qubit gates, 
into how many layers of gates can the sequence be decomposed so that no two gates act on the same qubit.  To construct the parallelized circuit, one keeps adding gates to the current level until there is a gate that shares a qubit with a previously added gate in that level, in which case you create a new level and continue.
We show that parallelizing a size $n$ random circuit results in a depth of $O(\log n)$ with high probability.

In order to avoid this overhead, we also consider a circuit model which is parallelized by construction. In this second model, a random maximum matching of on the complete graph is chosen and a random two-qubit gate is applied to qubits that are joined by an edge. We will also be interested in partially parallelized construction when in each time step a random edge is drawn from each of a set of coarse grained cells on a $d$-dimensional lattice.

A model of random circuits of a certain size defines a measure over unitary transformations on $n$ qubits that we call $p_{\circuit}$. We will sometimes compare the behaviour of the circuit to a unitary transformation chosen from the Haar measure $p_{\haar}$ over the full unitary group on $n$ qubits.

As mentioned earlier, the second-order moment operator will play an important role in all our proofs. The second-order moment operator is a super-operator acting on two copies of the space of operators acting on the ambient Hilbert space, which is an $n$-qubit space in our setting. For a measure $p$ over the unitary group, we can define the second moment operator $M$ as 
\[
M[X \otimes Y] = \exc{U \sim p}{UXU^\dagger \otimes UYU^\dagger}.
\]
%
% For next version, add some intuition why this is called second moment
% i.e., add link to polynomials of degree 2
%
%Let  $p(U)dU$ be a distribution over unitary transformations $U$ acting on a Hilbert space $\mathcal{H}$.  Let $UXU^\dagger$ be the linear map $A$ on the space of operators on $\mathcal{H}$, with matrix elements $a_{XX'}$. The operator, $$M[X\otimes Y] =\int dUp(U) UXU^\dagger \otimes UYU^\dagger,$$ contains all 2nd order moments, $\int dUp(U) a_{XX'} a_{YY'}$, 
%of $A$, and is referred to as the 2nd-order moment operator of the distribution.   
%
%
%
In particular $M_{\haar} = \exc{U \sim p_{\haar}}{UXU^\dagger \otimes UYU^\dagger}$.
Any distribution for which $M = M_{\haar}$ is referred to as a two-design.  We will be using the following properties of the second moment operator.
%have been shown to hold \cite{HL09, BHH12}. 
\begin{itemize}
\item For a circuit composed of $t$ gates chosen independently, the second moment operator is $M_{\circuit} = M^t$ where $M$ is the second-order moment operator corresponding to the measure obtained when applying one gate.
%The 2nd order moment operator $M^{(k)}$ for a distribution $p_k(U)dU$ generated by sampling $p(U)dU$ independently $k$-times and constructing the unitary, $U=U_k \ldots U_1$ (i.e. the k-fold convolution of $p(U)dU$)  is given by $M^{(k)} = M^k$.
\item For all the measures $p$ we consider here, the second moment operator is Hermitian. In addition, all eigenvalues of the moment operator are bounded in absolute value by $1$ and for the measures we consider, $1$ is the only eigenvalue of magnitude $1$.
%has support on a universal subset of $U(N)$, then $|\lambda|\le 1$ for each eigenvalue of $M$, with $|\lambda|=1$ only for the space of fixed points $M[X]=X$.
%\comment{I am not sure what you are trying to say in this point... Isn't always the case that the eigenvalues are bounded in absolute value by 1? Also is the purpose of the second point to say that $-1$ cannot be an eigenvalue?}
\item The eigenspace $\cV$ for the eigenvalue $1$ can be shown to be the space of operators $X$ acting on $2n$ qubits such that for all unitary transformations $U\otimes U X U^\dagger \otimes U^\dagger = X$ (for the distributions $p$ we consider here). It follows that this space is the span of the identity operator and the swap operator.
%The space of fixed points, $V$ of $M$ is given by the space of operators $X$ that are invariant under all unitary transformations of the form $U\otimes U X U^\dagger \otimes U^\dagger= X$.  It follows that the invariant space is spanned by the identity $\1$ and swap permutation $\mathbb{S}$.  The two-qubit swap is written $\mathbb{S}=+\sigma_o\sigma_o + \sigma_x\sigma_x+\sigma_y\sigma_y+\sigma_z\sigma_z$, and the $2n$-qubit swap is written $\mathbb{S}^{\otimes n}$.  Thus,  the invariant space consists of the uniform linear combination of two copies of each Pauli string.    
%\comment{Do you mean that all the weight is the same on all non-zero Pauli strings?}
The moment operator of the Haar measure, $M_{\haar}$, is the projector onto $\cV$.
%\item  $M$ is frustration free.  That is if $M\ket{A}=\ket{A}$ then $m^{ij} \ket{A}=\ket{A}$, $\forall i,j$.  This follows since $\ket{A}$ is invariant under $U\otimes U $ for all $U$ in $U(N)$.

%\item Under the condition, $\int dU p(U) \mbox{tr}[( X \otimes X )(UX'U^\dagger\otimes  UX''U^\dagger)]= 0$ unless $X'=X''$,  it follows that  $M$ is closed on the space spanned by $X\otimes X$ on which it is a bistochastic matrix.  In this case the evolution of the second order moments  are completely described by  a Markov chain which can be interpreted as a random walk over Pauli strings.  A sufficient condition for this is that $p(U)dU$ is invariant under all local unitary transformations.  We will limit consideration to such distributions. 

%\item Under the condition,  $p(U)dU = p(U^\dagger)dU$, the moment operator is Hermitian.  Such random quantum circuits will be referred to as reversible in analogy to classical Markov chains. 
\end{itemize} 

These properties imply that if $M$ is the second-order moment operator associated to the random quantum circuits we consider here, $M^t$ converges to $M_{\haar}$ as $t\rightarrow\infty$ at an asymptotic rate determined by the second largest $\lambda_2$ (in absolute value) eigenvalue of $M$. The gap of the moment operator is defined by $\Delta = 1 - \lambda_2$, and the larger the gap, the faster the moment operator converges to $M_{\haar}$. In order to study random quantum circuits when the interaction graph is a $d$-dimensional lattice, we will need a lower bound on the gap of these random quantum circuits. In order to obtain that, we proceed as in \cite{Znidaric2, BHH12} seeing the second-order moment operator as a local Hamiltonian. In fact, for a sequential random quantum circuit, we can write the second-order moment operator as follows: 
$$M_{\circuit} = \sum_{i<j} q_{ij} m_{ij},$$
where $m_{ij}=\exc{U \sim \tilde{p}} {UXU^\dagger \otimes UYU^\dagger} $ and $\tilde{p}$ is the normalized measure over gates that act only on qubits $i$ and $j$ and $q_{ij}$ is the total probability over such gates. Then, one can use a result on the gap of local  frustration free Hamiltonians \cite{Nachtergaele}. The property of being frustration free in this context simply mean that if $X$ is invariant for $M_{\circuit}$, then it is also invariant for the terms $m_{ij}$, which follows easily from the properties mentioned above.

As mentioned earlier, the Pauli basis will play an important role in our analysis. Consider a representation of the moment operator $M$ in the basis defined by $\sigma_{\nu} \otimes \sigma_{\mu}$, with $\nu, \mu \in \{0,1,2,3\}^n$. This defines a matrix $\{Q\left((\mu, \mu'), (\nu, \nu')\right)\}_{\mu, \mu', \nu, \nu'}$ of size $16^n \times 16^n$.

First, it can be shown that for the random quantum circuits we consider here, we have
\[
Q((\mu, \mu'), (\nu, \nu')) = \frac{1}{4^n} \exc{U \sim p}{ \mbox{tr}[( \sigma_\nu \otimes \sigma_{\nu'} ) (U\sigma_{\mu}U^\dagger\otimes  U\sigma_{\mu'}U^\dagger)]}=0,
\] 
unless $\nu=\nu'$ and $\mu=\mu'$. As a result, we will simply write $Q(\mu, \nu)$. Note that for any $\mu$, we have
\begin{align*}
\sum_{\nu \in \{0,1,2,3\}^n} Q(\mu, \nu) &= \sum_{\nu \in \{0,1,2,3\}^n} \frac{1}{4^n} \exc{U \sim p}{ \mbox{tr}[( \sigma_\nu \otimes \sigma_{\nu} ) (U\sigma_{\mu}U^\dagger\otimes  U\sigma_{\mu}U^\dagger)] } \\
&= \frac{1}{4^n} \exc{U \sim p}{\sum_{\nu \in \{0,1,2,3\}^n}  \mbox{tr}[\sigma_\nu U\sigma_{\mu}U^\dagger ]^2 } \\
&= \frac{1}{2^n} \exc{U \sim p}{\tr\left[ \left(U \sigma_{\mu} U^{\dagger}\right)^2 \right]} \\
&= \frac{1}{2^n} \tr[\sigma_{\mu}^2] \\
&= 1.
\end{align*}
This proves that $\{Q(\mu, \nu)\}_{\mu, \nu}$ can be seen as the transition matrix of a Markov chain on the set of Pauli strings $\{0,1,2,3\}^n$. This Markov chain is going to play an important role throughout the paper and the information theoretic properties we are interested in are going to be expressed in terms of its properties. More precisely, we are going to consider the chain obtained by removing the state $0^n$ (which is isolated from the rest of the chain).

Given a set of interacting pairs, to which Haar random (or Clifford) two-qubit gates are applied,  the Markov chain for the sequential random quantum circuit is constructed in the following way.  Consider only gates acting on qubits $i$ and $j$. If the Pauli string is  $\sigma_0 \otimes \sigma_0$ on $i$ and $j$, then there are no transitions induced as the gate acts as the identity on the string.  If the value of the string is  $\sigma_{a} \otimes \sigma_{b}$ with $(a,b) \neq (0,0)$ on qubits $i$ and $j$, then following from invariance of the Haar measure under unitary transformations, the value of the string on $i$ and $j$ transitions to each  $\sigma_{c} \otimes \sigma_{d}$ where $(c,d)$ is chosen uniformly from $\{0,1,2,3\}^2 - \{(0,0)\}$.  An average over all such transition for each interacting pair allowed by the interaction graph results in the Markov chain for this circuit. 
%Note that the uniform distribution over two-qubit Clifford gates generate precisely the same transition probabilities. 

%In fact, this is the definition of a random two-qubit Clifford gate.

%All what was used before is the gap and from that you cannot get anything better than $\Theta(n^2)$. One way to see this is that the gap does not depends on the geometry of the interaction graph, for any connected graph the gap will be $c/n$. And for the path $O(n^2)$ steps are needed for mixing.

%Here we study the properties of interest directly without going through the gap.

\section{Strong scrambling, decoupling and quantum error correction}
In this section, we prove that a random circuit of size $O(n \log^2 n)$ scrambles (on average over the choice of circuit) any initial state in the sense that all subsets of size at most $f n$ for some constant $f$ are very close to maximally mixed. In order to prove such a result, we consider the total mass of the coefficients corresponding to the Pauli strings of weight at most $fn$, and prove that it is small with very high probability. The common thing between proving strong scrambling and obtaining error correcting codes with large minimum distance is that the probability bounds should be close to optimal. More precisely, we prove that a Pauli string of weight $\ell$ is mapped by the random circuit to a Pauli string of weight at least $f n$ with probability at least roughly $1 - \frac{1}{\binom{n}{\ell}}$. The following Section \ref{sec:parallelization} then says that this circuit can with high probability be parallelized so that it has depth $O(\log^3 n)$. In the following sections, we see how we can interpret our upper bounds on the total mass on low-weight Pauli strings to prove results on decoupling and quantum error correcting codes for low-depth random quantum circuits.
%This section deals with a strong notion of scrambling that requires 
\subsection{Sequential random circuit}
\begin{theorem}
\label{thm:strong-scrambling-seq}
Let $\rho(0)$ be an initial arbitrary mixed state on $n$ qubits and $\rho(t)$ be the corresponding state after the application of $t$ random two-qubit gates (the sequential circuit model). Then provided $f$ is such that $f \log 3 + h(f) - \frac{\log 3}{2} < 0$ and $t > c n \log^2 n$ (for some large enough constant $c$), we have for all subsets $S$ of size at most $fn$,
\begin{equation}
\label{eq:purity-statement}
\ex{ \tr[\rho_S(t)^2] } \leq \frac{1}{2^{|S|}} + \frac{1}{2^{|S|} \poly(n)},
\end{equation}
where the expectation is taken over the random circuit. This implies that
\begin{equation}
\label{eq:l1-statement}
\ex{ \left\| \rho_S(t) - \frac{\1}{2^{|S|}} \right\|^2_1} \leq \frac{1}{\poly(n)}. 
\end{equation}
\end{theorem}
\begin{proof}
First, observe that \eqref{eq:purity-statement} easily implies \eqref{eq:l1-statement} using the Cauchy-Schwarz inequality:
\begin{align*}
\left\| \rho_S(t) - \frac{\1}{2^{|S|}} \right\|^2_1 &\leq 2^{|S|}  \left\| \rho_S(t) - \frac{\1}{2^{|S|}} \right\|^2_2 \\
&= 2^{|S|} \left( \tr[\rho_S(t)^2] - 2 \frac{\tr[ \rho_S(t) ]}{2^{|S|}} + \frac{\tr[\1]}{2^{2|S|}} \right) \\
&= 2^{|S|} \tr[ \rho_S(t)^2 ] - 1.
\end{align*}
To compute $\tr[ \rho_S(t)^2 ]$, we decompose $\rho_S(t)$ in the Pauli basis: 
\[
\rho_S(t) = \sum_{\nu \in \{0,1,2,3\}^{S} } 2^{-|S|}\tr[\sigma_{\nu} \rho(t)] \sigma_{\nu}.
\]
As a result, we have
\begin{align*}
\tr [\rho_S^2(t)] &= \sum_{\nu \in \{0,1,2,3\}^S} \frac{\tr[\sigma_{\nu} \rho_S(t) ]^2}{2^{|S|}} \\
			&= \sum_{\nu \in \{0,1,2,3\}^S} \frac{\tr[\sigma_{\nu} \otimes \1_{S^c} \tr_{S^c} [\rho(t)] ]^2}{2^{|S|}} \\
			&= \frac{1}{2^{|S|}} + \sum_{\nu \in \{0,1,2,3\}^S, \nu \neq 0} \frac{\tr[\sigma_{\nu} \otimes \1_{S^c} \rho(t) ]^2}{2^{|S|}} \\
			&\leq \frac{1}{2^{|S|}} +  \sum_{\nu : 1 \leq w(\nu) \leq |S|} \frac{\tr[\sigma_{\nu} \rho(t) ]^2}{2^{|S|}}.
\end{align*}
%As a result, we get
%\begin{equation}
%\label{eq:bound1norm}
%\left\| \rho_S(t) - \frac{\1_S}{2^{|S|}} \right\|^2_1 \leq \sum_{\nu : 1 \leq w(\nu) \leq |S|} \tr[\sigma_{\nu} \rho(t) ]^2.
%\end{equation}

Our objective now is to study the evolution of the quantity $\ex{\sum_{\nu : 1 \leq w(\nu) \leq |S|} \tr[\sigma_{\nu} \rho(t) ]^2}$ as a function of $t$. As we described in the preliminaries, applying a random two-qubit gate has a simple effect on the decomposition into the Pauli basis: an identity on two qubits always gets mapped to an identity and a non-identity Pauli string on two qubits gets mapped to a uniformly chosen non-identity Pauli string (of which there are $15$). 

Our focus will be to study the Markov chain that describes the evolution of the distribution of the weight of the different levels $\sum_{\nu : w(\nu) = k} \ex{\tr[\sigma_{\nu} \rho(t)]^2}$.
%In fact, it turns out that applying a random two-qubit gate, the Pauli operators on these two-qubits get transformed as follows: the identity always gets mapped to the identity and non-identity Pauli strings get mapped to a uniformly chosen non-identity Pauli string (of which there are $15$). For a proof of this, we refer the reader to \cite{HL09} (probably was done before, what's the right reference).
%The above calculation shows that the vector $\{\ex{\tr[\sigma_{\nu}^{\otimes 2} \rho(t)^{\otimes 2}]} \}_{\nu}$ is obtained by applying some fixed matrix (independent of time) that to the vector $\{\ex{\tr[\sigma_{\mu}^{\otimes 2} \rho(t-1)^{\otimes 2}]} \}_{\mu}$ at time $t-1$.
%TODO: this needs more detail, you have to look at which qubits is $U_t$ acting on to be precise. Will do it later.
%Recall that our objective is to study the evolution of the distribution of weights assigned to each Pauli string. 
More precisely, we can write for any $k \in \{1, \dots, n\}$,
\begin{align*}
\sum_{\nu: w(\nu) = k}\ex{\tr[\sigma_{\nu} \rho(t)]^2} 
&= \sum_{\nu: w(\nu) = k, \mu} \ex{\tr[\sigma_{\mu} \rho(t-1)]^2} \ex{\tr[\sigma_{\nu} U_t \sigma_{\mu} U_t^{\dagger}]^2} \\
&= P(k-1, k) \sum_{\mu: w(\mu) = k-1}\ex{\tr[\sigma_{\mu}\rho(t-1)]^2} \\
& + P(k, k) \sum_{\mu: w(\mu) = k}\ex{\tr[\sigma_{\mu} \rho(t-1)]^2} \\
& + P(k+1, k) \sum_{\mu: w(\mu) = k+1}\ex{\tr[\sigma_{\mu} \rho(t-1)]^2},
\end{align*}
where the matrix $P \in \RR^{n \times n}$ is defined by
\[
P(x,y) = \left\{
\begin{array}{cc}
1- \frac{2x(3n-2x-1)}{5n(n-1)} & \text{ if } y=x \\
\frac{2x(x-1)}{5n(n-1)} & \text{ if } y = x - 1\\
\frac{6x(n-x)}{5n(n-1)} & \text{ if } y = x + 1\\
0 & \text{ otherwise.}
\end{array} \right.
\]

We refer the reader to \cite{HL09} for more details on how to derive the parameters of this Markov chain. In fact, \cite{HL09} study the mixing time of this Markov chain. Here, we need to analyze a slightly different property: starting at some point $\ell$, what is the probability that after $t$ steps the random walk ends up in a point $\leq fn$? One can obtain bounds on this probability using the mixing time but these bounds only give something useful for our setting if $t = \Omega(n^2)$. So we will need to improve the analysis of \cite{HL09} and go directly for computing the desired probability instead of going through the mixing time. More precisely, by defining the Markov chain $\{X_s(\ell)\}_{s \geq 0}$ that starts at $\ell$ and has transition probabilities given by $P$, we have
\begin{equation}
\label{eq:pauli-to-mc}
\sum_{k=1}^{fn} \sum_{\nu: w(\nu) = k} \ex{\tr[\sigma_{\nu} \rho(t)]^2}
= \sum_{\ell=1}^{n} \sum_{\nu: w(\nu) = \ell} \tr[\sigma_{\nu} \rho(0)]^2 \pr{X_t(\ell) \leq fn}.
\end{equation}
% Thus, it suffices to study the Markov chain defined by $P$ after $t$ steps starting at the distribution defined by the normalized version of the vector $\{\sum_{\nu: w(\nu) = k} \tr[\sigma_{\nu} \rho(0)]^2\}_k$. 
%Recall that what we want is to bound \eqref{eq:bound1norm}. We see the vector $\{\sum_{\nu: w(\nu) = k} \tr[\sigma_{\nu} \rho(t)]^2\}_k$ as a distribution and we study the evolution of it. 
%Our objective is to show that
%\[
%\sum_{k=1}^{s} \sum_{\nu: w(\nu) = k} \ex{\tr[\sigma_{\nu} \rho(t)]^2} \leq \e.
%\]
If the initial state $\rho(0)$ is a pure product state, then one can verify that
\[
\sum_{\nu: w(\nu) = \ell} \tr[\sigma_{\nu} \rho(0)]^2 = \binom{n}{\ell}.
\]
%
% This is a proof that for product states we get at least \binom{n}{\ell}
%
%If $\rho(0)$ is a product state, then we can compute the expressions $\sum_{\nu: w(\nu) = k} \tr[\sigma_{\nu} \rho(0)]^2\}_k$ for $k \in \{1, \dots, n\}$.
%We have for any subset $S$ of size $n-k$,
%\begin{align*}
%\sum_{\nu: \nu_S = 0^{n-k}, \nu_{S^c} \in \{1,2,3\}^{S^c}} \tr[\sigma_{\nu} \rho(0)]^2
%&= \prod_{i \in S} \tr \left[ \rho_i \right]^2 \sum_{\mu \in \{1,2,3\}^{S^c}} \prod_{j \in S^c}\tr[ \sigma_{\mu_j} \rho_j ]^2 \\
%&= 1 \cdot \prod_{j \in S^c}\left( \tr[ \sigma_{1} \rho_j ]^2 + \tr[ \sigma_{2} \rho_j ]^2 + \tr[ \sigma_{3} \rho_j ]^2 \right) \\
%&= \prod_{j \in S^c} \left(2\tr[\rho^2_j] - 1\right).
%\end{align*}
%If $\rho(0)$ is pure, the previous expression evaluates to $1$ and the starting distribution for the Pauli weights is given by 
%\[
%\sum_{\nu: w(\nu) = k} \tr[\sigma_{\nu} \rho(0)]^2 = \binom{n}{k}.
%\]
%\comment{For which class of states can we determine this distribution? For a maximally entangled state it gives a binomial on the even values for example, does this state scramble quickly?}
%
% End of proof
%
In general, we have
\begin{align*}
\sum_{\nu: w(\nu) = \ell} \tr[\sigma_{\nu} \rho(0)]^2 &\leq \sum_{S : |S| = \ell} 2^{|S|}\tr[ \rho_S(0)^2 ] \\
&\leq 2^{\ell} \binom{n}{\ell}.
\end{align*}
\comment{Could one obtain a better bound of just $\binom{n}{\ell}$?}

The main technical result in this proof is in Theorem \ref{thm:mc-convergence} (which we defer to the appendix), where we obtain a bound on $\pr{X_t(\ell) \leq fn)} \leq \frac{2^{f \log 3 + h(f)}}{\binom{n}{n/2} 3^{n/2} } + \frac{1}{2^{\ell} \binom{n}{\ell} \poly(n)}$. Plugging this into \eqref{eq:pauli-to-mc}, we obtain
\begin{align*}
\sum_{k=1}^{fn} \sum_{\nu: w(\nu) = k} \ex{\tr[\sigma_{\nu} \rho(t)]^2}
&\leq \sum_{\ell=1}^n \sum_{\nu: w(\nu) = \ell} \ex{\tr[\sigma_{\nu} \rho(0)]^2} \cdot \frac{1}{2^{\ell} \binom{n}{\ell} \poly(n)} + \frac{2^{f \log 3 + h(f)}}{\binom{n}{n/2} 3^{n/2} } \sum_{\ell=1}^n \sum_{\nu: w(\nu) = k} \ex{\tr[\sigma_{\nu} \rho(0)]^2}\\
&\leq \sum_{\ell=1}^n 2^{\ell} \binom{n}{\ell} \cdot \frac{1}{2^{\ell} \binom{n}{\ell} \poly(n)} + 2^n \cdot \tr[\rho(0)^2] \cdot  \frac{2^{f \log 3 + h(f)}}{\binom{n}{n/2} 3^{n/2} }\\
&\leq \frac{1}{\poly(n)},
\end{align*}
provided $f$ is such that $f \log 3 + h(f) - \frac{\log 3}{2} < 0$.
%Our problem now reduces to studying the distribution of $X_t$ where $\{X_s\}_{s \geq 0}$ is a Markov chain with transition matrix $P$ and where $X_0$ is defined by $\pr{X_0 = k} = \frac{\binom{n}{k}}{2^n - 1}$ for all $k \in \{1, \dots, n\}$. Theorem \ref{thm:mc-convergence} proves that we have 
%\[
%\pr{X_t \leq fn} \leq \e 2^{-n}.
%\]
%This implies that assuming $\rho(0)$ is a pure product state
%\begin{align*}
%\sum_{k=1}^{fn} \sum_{\nu: w(\nu) = k} \ex{\tr[\sigma_{\nu} \rho(t)]^2} 
%&= (2^n - 1) \cdot \sum_{k=1}^{fn} \sum_{\nu: w(\nu) = k} \ex{\frac{\tr[\sigma_{\nu} \rho(t)]^2}{2^{n} - 1}} \\
%&= (2^n - 1) \cdot \pr{X_t \leq fn} \\
%&\leq \e.
%\end{align*}
%We now study the chain defined by $P$.
%\begin{claim}
%Let $X_0$ have distribution $\pr{X_0 = k} = \frac{\binom{n}{k}}{2^n - 1}$. Let $f < 1/2$.
%There is a $c$ such that for $t \geq c n \log(n/\e)$ such that
%\[
%\pr{X_t \leq f n} \leq \e 2^{-n}.
%\]
%\end{claim}
%\begin{proof}
%We have
%\begin{align*}
%\pr{X_t \leq f n} &\leq \sum_{\ell = 1}^{n} \pr{X_0 = \ell} \pr{X_t(\ell) \leq fn} \\
%&\leq \sum_{\ell = 1}^{n} \frac{ \binom{n}{\ell} }{2^n - 1} \pr{X_t(\ell) \leq fn} 
%\end{align*}
%\end{proof}
%%end of proof of claim
\end{proof}

\subsection{Parallelizing the circuit}
\label{sec:parallelization}
Recall that we are interested in the depth of random circuits. A priori, the circuit studied in the previous section has a depth that is as large as the number of gates which is nearly linear. But in general in such a circuit there are many successive gates that are applied on disjoint qubits so they could be actually performed in parallel. More precisely, we look at the gates one by one in the order they are applied. For the purpose of this section, the gates can simply be labelled by  the two qubits the gate acts upon. To construct the parallelized circuit, one keeps adding gates to the current level until there is a gate that shares a qubit with a previously added gate in that level, in which case you create a new level and continue. In the following proposition, we prove that by parallelizing a random circuit on $n$ qubits having $t$ gates we obtain with high probability a circuit of depth $O(\frac{t}{n} \log n)$.

\begin{proposition}
\label{prop:parallelization}
Consider a random sequential circuit composed of $t$ gates where $t$ is a polynomial in $n$. Then parallelize the circuit as described above. Except with probability $1/\poly(n)$, you end up with a circuit of depth at most $O\left(\frac{t}{n} \log n\right)$.
\end{proposition}
In order to prove this lemma, we use the following calculation:
\begin{lemma}
Let $G_1, \dots, G_k$ be a sequence of independent and random gates $G_i \in \binom{n}{2}$, then the probability that $G_1, \dots, G_k$ form a circuit of depth $k$ is at most $\left(\frac{2}{n} \right)^{k-1} \cdot k !$
\end{lemma}
\begin{proof}
We prove this by induction on $k$. For $k = 2$, we may assume $G_1 = (1,2)$, in which case $\pr{G_2 \cap \{1,2\} \neq \emptyset} \leq 4/n$.
Now the probability that $G_1, \dots, G_{k+1}$ form a circuit of depth $k+1$ can be bounded by
\[
\pr{G_1, \dots, G_k \text{ form a circuit of depth $k$ } } \cdot \pr{G_{k+1} \cap \left( G_1 \cup \cdot \cup G_k \right) \neq \emptyset | G_1, \dots, G_k \text{ form a circuit of depth $k$ } }.
\]
Now it suffices to see that, conditioned on $\event{G_1, \dots, G_k \text{ form a circuit of depth $k$}}$, the number of nodes occupied by $G_1, \dots, G_k$ is at most $k+1$. Thus, using this fact and the induction hypothesis, we obtain
a bound of 
\[
\left(\frac{2}{n}\right)^{k-1} k! \cdot 2 \cdot \frac{k+1}{n} = \left( \frac{2}{n} \right)^{k} (k+1)! \ ,
\]
which conclude the proof.
\end{proof}

\begin{proof}[of Proposition \ref{prop:parallelization}]
Suppose we apply $m$ gates for some $m$ to be chosen later.
\begin{align*}
\pr{G_1, \dots, G_m \text{ form a circuit of depth at least $d$ } } 
&= \pr{ \exists (i_1, \dots, i_d) \in [m]^d : G_{i_1}, \cdots, G_{i_d} \text{ form a circuit of depth d} } \\ 
&\leq \binom{m}{d}  \left(\frac{2}{n} \right)^{d-1} \cdot d ! \\
&\leq m^d \cdot \left(\frac{2}{n} \right)^{d-1}.
\end{align*}
Now we can fix $m = n/4$ and $d = c \log n + 1$ for some constant $c$ to be chosen depending on the desired probability bound, then we have
\begin{align*}
\pr{G_1, \dots, G_m \text{ form a circuit of depth at least $d$ } } 
&\leq m \cdot \left(\frac{2m}{n}\right)^{d-1} \leq n^{-c+1}.
\end{align*}
This proves that every set of $n/4$ gates generate a circuit of depth at most $c\log n + 1$ with probability at least $1-1/n^{-c+1}$, and so if we have $4t/n$ such sets, we get depth at most $4t/n(c \log n + 1)$ with probability at least $1-4t/n^{c}$.
\end{proof}

The next corollary follows directly from Theorem \ref{thm:strong-scrambling-seq} and Proposition \ref{prop:parallelization}.
\begin{corollary}
In the parallelized random quantum circuit model with depth $O(\log^3 n)$, we have 
\begin{equation}
\ex{ \left\| \rho_S(t) - \frac{\1}{2^{|S|}} \right\|^2_1} \leq \frac{1}{\poly(n)}
\end{equation}
for all subsets $S$ of size at most $fn$ with $f$ such that $f \log 3 + h(f) - \frac{\log 3}{2} < 0$.
\end{corollary}

\subsection{Decoupling and quantum error correcting codes}

Scrambling is related to the notion of decoupling. The idea of decoupling plays an important role in quantum information theory and many coding theorems amount to proving a decoupling theorem \cite{HOW05, HOW06, HHYW08, ADHW09, Dup09, DBWR10}.

Consider the setting described in Figure \ref{fig:decoupling}. Let $\ket{\Phi}_{MM'}$ and $\ket{\psi}_{AA'}$ be pure states on $MM'$ and $AA'$ respectively. Then apply some unitary transformation to the system $M'A'$ (which for us is going to be a random quantum circuit) and map it to a system that we call $B$. Let us denote by $\ket{\rho}_{BMA}$ the output state. Assume now that the reduced state $\rho_{MS}$ on $M$ together with some subset $S$ of the qubits of $B$ is a product state: $\rho_{MS} = \rho_M \otimes \rho_S$ (the subsystem $S$ is \emph{decoupled} from the reference $M$). Then by Uhlmann's theorem (or the unitary equivalence of purifications), there exists an isometry acting on $AS^c$ that recovers a purification of the system $M$. If for example $\ket{\Phi}_{MM'}$ is a maximally entangled state, then the previous argument shows that if we input quantum information into the $M'$ system, it can be recovered from the systems $AS^c$ alone with no need for the system $S$.

	\begin{figure}[h]
	\begin{center}
		\includegraphics[scale=0.8]{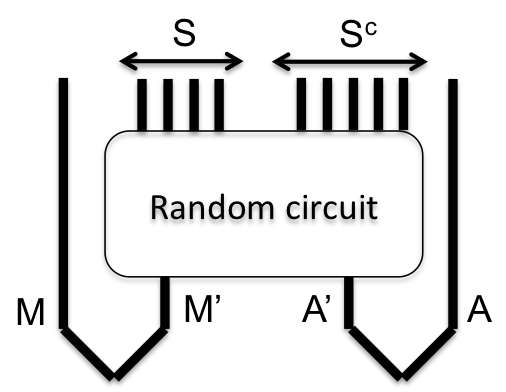}
		\caption{Illustration of decoupling for random quantum circuits}
		\label{fig:decoupling}
	\end{center}
	\end{figure}

In the following for simplicity, we focus on the case
%We will consider the maximally entangled state on $m$ qubits
\[
\Phi_{MM'} = \frac{1}{2^{2m}} \sum_{\nu \in \{0,1,2,3\}^m} \sigma_{\nu} \otimes \sigma_{\nu} ,
\]
where the systems $M$ and $M'$ consist of $m$ qubits. For the $AA'$ system, we will focus on two important cases: First where $A'$ is already in a pure state $\ket{\psi}_{A'} = \ket{0}_{A'}$, so that it can be written in the Pauli basis as
\begin{equation}
\label{eq:psi-0}
\psi_{A'} = \frac{1}{2^{n-m}} \sum_{\nu \{0, 3\}^{n-m}} \sigma_{\nu} \ ,
\end{equation}
and second the case where $\psi_{AA'}$ is maximally entangled so that
\begin{equation}
\label{eq:psi-max-entangled}
\psi_{AA'}  = \frac{1}{2^{2(n-m)}} \sum_{\nu \in \{0,1,2,3\}^{n-m}} \sigma_{\nu} \otimes \sigma_{\nu} ,
\end{equation}
which corresponds to entanglement assisted communication. Of course, one could obtain a statement for general states and this would depend on the decomposition of the states $\Phi_{MM'}$ and $\psi_{AA'}$ in the basis of Pauli strings and more precisely on the weight distribution of this decomposition.

A decoupling statement is very similar in spirit to scrambling and the analysis is almost the same except that we use a specific state for the input. Using Theorem \ref{thm:mc-convergence} which was the main technical ingredient in the proof of Theorem \ref{thm:strong-scrambling-seq}, we can get the following decoupling and coding results.

\begin{theorem}
\label{thm:decoupling}
In the setting of equation \eqref{eq:psi-0}, we have if $m < \beta n$ and any $S$ of size $|S| \leq f n$ with $\beta < 1/9$ and $\beta < \log3/2 - f \log 3 -  h(f)$, and $\rho$ is the state obtained after applying parallelized random quantum circuit of depth $O(\log^3 n)$, we have
\[
\left\| \rho_{MS} - \frac{\1}{2^m} \otimes \frac{\1}{2^{|S|}} \right\|_1 \leq \frac{1}{\poly(n)}
\]
with probability $1-1/\poly(n)$ over the choice of the circuit.

In the setting of equation \eqref{eq:psi-max-entangled} (entanglement assisted coding), we have for $\beta < 2/3$ and $\beta < \frac{1 + \log3/2 - f \log 3 -  h(f)}{2}$,
\[
\left\| \rho_{MS} - \frac{\1}{2^m} \otimes \frac{\1}{2^{|S|}} \right\|_1 \leq \frac{1}{\poly(n)}
\]
with probability $1-1/\poly(n)$ over the choice of the circuit.
%we can encode using a parallelized circuit of depth $O(\log^3 n)$ at a rate $\beta < 2/3$ and correct any $fn$ erasures with polynomially small error provided $2\beta + f \log 3 + h(f) - 1 - \log 3/2 < 0$.
\end{theorem}
\begin{remark}
The rates we obtain here are not optimal, but we prove that it is possible to code at constant rate with a constant fraction of errors using a circuit of polylogarithmic depth. It would be interesting to determine whether it is possible to achieve the capacity of the erasure channel using such shallow circuits.
For example, it would be interesting to improve the bound in the entanglement assisted case to $\beta < \frac{2 - f \log 3 -  h(f)}{2}$. This is the bound one would get for a random unitary distributed according to the Haar measure on the unitary group acting on $n$ qubits and is reminiscent of the entanglement assisted capacity of the depolarizing channel.
%The condition $f \log 3 + h(f) - 1 - \frac{\log3}{2} < 0$ can probably be slightly improved. It would be very interesting to determine whether it is possible to get all the way to $f \log 3 + h(f) - 2 < 0$, which is the best we can hope for as this is what one would get if we applied a random unitary on the full unitary group acting on the $n$-qubit space.
\end{remark}
\begin{proof}
We start with the entanglement assisted case, for which the calculation is a bit simpler. We apply a random quantum circuit to the system $M'A'$. We can write the initial state on $MM'A'$ as
\[
\rho(0) = \frac{1}{2^{2m}} \sum_{\nu \in \{0,1,2,3\}^m} \sigma_{\nu} \otimes \sigma_{\nu} \otimes \frac{\1}{2^{n-m}}.
\]
We study the evolution of the state $\rho(t)$ when we apply $t$ random gates, more precisely we study the behaviour of the reduced state when a subset size $k = (1-f) n$ qubits are discarded (from the $M'A'$ system), the objective is to show that the remaining state is close to maximally mixed. We have
\begin{align*}
\ex{\left\| \rho(t)_{MS} - \frac{\1}{2^{m + fn}} \right\|^2_1} 
&\leq \sum_{ \nu \in \{0,1,2,3\}^m, \mu \in \{0,1,2,3\}^{S}, (\nu, \mu) \neq 0} \ex{\tr[ \sigma_{\nu} \otimes \sigma_{\mu} \otimes \1 \rho(t) ]^2} \\
&\leq \sum_{ \nu \in \{0,1,2,3\}^m, w(\mu) \leq fn, (\nu, \mu) \neq 0 } \ex{\tr[ \sigma_{\nu} \otimes \sigma_{\mu} \rho(t) ]^2} \\
&= \sum_{ \nu \in \{0,1,2,3\}^m, \nu \neq 0 } \tr[\sigma_{\nu} \otimes \sigma_{\nu} \otimes \1 \rho(0)]^2 \pr{X_t(w(\nu)) \leq fn} \\
&= \sum_{ \ell =  1}^{m} \binom{m}{\ell} 3^{\ell} \pr{X_t(\ell) \leq fn},
\end{align*}
where we used the same notation as in the proof of Theorem \ref{thm:strong-scrambling-seq}: $X_t(\ell)$ is the random variable denoting the weight of the Pauli string obtained after applying $t$ random gates to the operator $\sigma_{\mu}$ for some $\mu$ of weight $\ell$.

If $t > c n \log^2 n$, we can apply Theorem \ref{thm:mc-convergence}, and obtain
\begin{align*}
\ex{\left\| \rho(t)_{MS} - \frac{\1}{2^{m + fn}} \right\|^2_1}
&\leq  \frac{1}{\poly(n)} \cdot \sum_{ \ell =  1}^{m} \binom{m}{\ell} 3^{\ell} \frac{1}{2^{\ell}\binom{n}{\ell}} + (4^m - 1) \cdot \frac{2^{f \log 3 + h(f)}}{\binom{n}{n/2} 3^{n/2}} \\
&= \frac{1}{\poly(n)} \cdot \sum_{ \ell =  1}^{m} \frac{m (m-1) \cdots (m-\ell+1)}{n(n-1) \cdots (n-\ell+1)} \left(\frac{3}{2}\right)^{\ell} + (4^m - 1) \cdot \frac{2^{f \log 3 + h(f)}}{\binom{n}{n/2} 3^{n/2}}.
\end{align*}
This means that provided $m$ and $f$ are small enough, most of the random circuits with $t = O(n \log^2 n)$ gates are good encoders that allow the (approximate) correction of any $fn$ erasure when entanglement assistance is available. In other words, if we write $m=\beta n$, then as long as $\beta < 2/3$ and $2\beta + f \log 3 + h(f) - 1 - \log 3/2 < 0$, the reference system $M$ is decoupled from any subset of at most $fn$ qubits of the output.

We now move to the case where the state $\ket{\psi}_{A'}$ is pure. In this case the initial state on $MM'A'$ can be written as
\[
\rho(0) = \frac{1}{2^{2m}} \cdot \frac{1}{2^{n-m}} \sum_{\nu \in \{0,1,2,3\}^m, \mu \in \{0,3\}^{n-m}} \sigma_{\nu} \otimes \sigma_{\nu} \otimes \sigma_{\mu}. 
\]
Then, the analysis is the same
\begin{align}
\ex{\left\| \rho(t)_{MS} - \frac{\1}{2^{m + fn}} \right\|^2_1} 
%&\leq \sum_{ \nu \in \{0,1,2,3\}^m, \mu \in \{0,1,2,3\}^{S}, (\nu, \mu) \neq 0} \ex{\tr[ \sigma_{\nu} \otimes \sigma_{\mu} \otimes \1 \rho(t) ]^2} \\
%&\leq \sum_{ \nu \in \{0,1,2,3\}^m, w(\mu) \leq fn, (\nu, \mu) \neq 0 } \ex{\tr[ \sigma_{\nu} \otimes \sigma_{\mu} \rho(t) ]^2} \\
&= \sum_{ \nu \in \{0,1,2,3\}^m, \mu \in \{0,3\}^{n-m}, , (\nu, \mu) \neq 0 } \tr[\sigma_{\nu} \otimes \sigma_{\mu} \rho(0)_{M'A'}]^2 \pr{X_t(w(\nu)+w(\mu)) \leq fn} \notag \\
&= \sum_{\ell=1}^n \left( \sum_{p=0}^{\min(\ell,m)} \binom{m}{p} 3^{\ell} \binom{n-m}{\ell - p} \right) \pr{X_t(\ell) \leq fn} \\
&\leq \sum_{\ell=1}^n \left( \sum_{p=0}^{\min(\ell,m)} \binom{m}{p} 3^{\ell} \binom{n-m}{\ell - p} \right) \frac{1}{\binom{n}{\ell} 2^{\ell} \poly(n)} + 2^{m+n} 
\frac{2^{f \log 3 + h(f)}}{\binom{n}{n/2} 3^{n/2}}.
 \label{eq:bound-ancilla}
%&= \sum_{ \ell =  1}^{m} \binom{m}{\ell} 3^{\ell} \pr{X_t(\ell) \leq fn},
\end{align}
If $m = \beta n$, the second term vanishes provided $\beta + f \log 3 + h(f) - \log3/2 < 0$. For the first term, we need to analyze carefully the number of Pauli strings $\sigma_{\nu} \otimes \sigma_{\mu}$ with a given weight $\ell$ which does not have an expression that is as simple as in the entanglement assisted case. Our objective is to prove that
\begin{align*}
\sum_{p=0}^{\min(\ell,m)} \binom{m}{p} 3^{\ell} \binom{n-m}{\ell - p} \leq \ell \binom{n}{\ell} 2^{\ell}.
\end{align*}
using the fact that $m$ is not too large, so that we get a vanishing bound on the trace distance. We bound the terms for $p \leq \ell/2$ and $p \geq \ell/2$ separately. We have
\begin{align*}
\sum_{p=0}^{\ell/2} \binom{m}{p} 3^{\ell} \binom{n-m}{\ell - p}
&\leq 3^{\ell/2} \sum_{p=0}^{\ell/2} \binom{m}{p} \binom{n-m}{\ell - p} \\
&\leq 2^{\ell} \binom{n}{\ell}.
\end{align*}
For $p \geq k/2$, this needs a bit more work. We have
\begin{align*}
\binom{m}{p} \binom{n-m}{\ell-p} = \frac{m (m-1) \cdots (m-p+1) \cdot (n-m) (n-m-1) \cdots (n-m-(\ell-p)-1)}{p! (\ell-p)!}.
\end{align*}
First we have $p! (\ell-p)! \geq ((\ell/2)!)^2 \geq \frac{\ell!}{2^{\ell}}$. We also have
\[
m(m-1) \cdot (m-p+1) \leq \left(\frac{m}{n} \right)^p n (n-1) \cdots (n-p+1).
\]
Moreover, as $p \leq m$, we have
\[
(n-m) \cdots (n - m - \ell + p + 1) \leq (n-p) \cdots (n - p - \ell + p + 1). 
\]
As a result, for $p \geq \ell/2$,
\[
3^p \binom{m}{p} \binom{n-m}{\ell-p} \leq 3^\ell \cdot 2^\ell \left(\frac{m}{n}\right)^{\ell/2} \cdot \binom{n}{\ell}.
\]
By choosing $m/n \leq 1/9$, we can now bound the number of Pauli strings of weight $\ell$:
\begin{align*}
\sum_{p=0}^{\min(\ell,m)} \binom{m}{p} 3^{\ell} \binom{n-m}{\ell - p}
\leq \ell 2^{\ell} \binom{n}{\ell}.
\end{align*}
Returning to \eqref{eq:bound-ancilla}, provided $m/n < 1/9$ we get
\begin{align*}
\ex{\left\| \rho(t)_{MS} - \frac{\1}{2^{m + fn}} \right\|^2_1} 
&\leq \sum_{\ell=1}^n \frac{\ell}{\poly(n)} + 2^{n+m} \cdot \frac{2^{f \log 3 + h(f)}}{\binom{n}{n/2} 3^{n/2}}.
\end{align*}
We can then obtain the claimed results by parallelizing these sequential circuits (Proposition \ref{prop:parallelization}).
\end{proof}
Another way of interpreting the analysis above is that random quantum circuits define good stabilizer codes, i.e., codes with a positive rate and linear distance. Such a result can be seen as a step towards understanding the complexity of encoding into good quantum error correcting codes. There are many results on various classical versions of this problem; see e.g., \cite{GHKPV12} for a recent result in this spirit.

\begin{theorem}[Good codes from low-depth circuits]
\label{thm:good-codes}
There exists non-degenerate stabilizer codes with encoding circuits of depth $O(\log^3 n)$ encoding $\beta n$ qubits into $n$ qubits and having a minimum distance $\alpha n$ for some constants $\alpha, \beta > 0$.
\end{theorem}
\begin{proof}

\iffalse
\begin{proof}
A circuit composed of Clifford gates that maps all Pauli strings of the form $\sigma_{\nu} \otimes \sigma_{\mu}$ with $\nu \in \{0,1,2,3\}^m$ and $\mu \in \{0,3\}^{n-m}$ into Pauli strings of weight at least $fn$ defines a stabilizer code encoding $m$ qubits and having distance $fn$.  For a Clifford circuit either 0That is exactly what the analysis in the proof of Theorem \ref{thm:decoupling} shows.
\end{proof}
\fi

It is not hard to see that a circuit composed of Clifford gates that maps all Pauli strings of the form $\sigma_{\nu} \otimes \sigma_{\mu}$ with $\nu \in \{0,1,2,3\}^m$ and $\mu \in \{0,3\}^{n-m}$ into Pauli strings of weight at least $fn$ defines a stabilizer code encoding $m$ qubits and having distance $fn$. That is exactly what the analysis in the proof of Theorem \ref{thm:decoupling} shows.
\end{proof}

%Decoupling when you remove systems of size at most $fn$. If works for EPR pairs, and I think it should work for other states as well. First without entanglement assistance and this gives good codes for the erasure channel and also good codes with half the minimum distance. Then I can also say something about the entanglement assisted case.

%An bound on the fraction of the circuits of a given length for which this is true should be added if possible.

The results in this section involve random quantum circuits of depth $O(\log^3 n)$. It would be interesting to improve these scrambling times to $O(\log n)$ instead. Our second set of results presented in the following section proves a weaker notion of scrambling in depth $O(\log n)$. This notion of scrambling is particularly relevant in the study of the black hole information paradox question. In addition, we also consider this notion of scrambling when the interaction graph is a $d$-dimensional lattice. 

\section{Scrambling and the black hole information paradox}
Here we show that there are natural random quantum circuit models that perform good entanglement assisted codes for the erasure channel.   Note that this is a weaker notion of scrambling since we require that only a constant number of initial low weight Pauli strings are brought to linear weight strings with high probability.   In this section we will consider a different model of random quantum circuit where gates are selected from among sets of matchings between neighbors on lattices of fixed dimension and from the complete graph.  For the $d$-dimensional models, to aid in our proofs we introduce an additional set of coarse grained blocks of size $O(\log n)$ between which disallow gates to be performed for coarse time steps of $O(\log^2 n)$. We show that a constant size message for typical random quantum circuits of depth $n^{1/d} \log^2 n$ and $\log n$ for random quantum circuits with gates that act on a bounded number of qubits between neighbors on $d$-dimensional lattices and two-qubits gates on the complete (infinite dimensional) graph.  For quantum circuits consisting of gates of fixed weight a straightforward upper bound is given by the radius of the interactions graphs of  $n^{1/d}$ and $\log n$, so that our results are essentially optimal.  Assuming the random quantum circuit models accurately capture the scaling behavior of typical Hamiltonians with the same interaction graph,  as argued in \cite{HP07} these time scales determine the time at which a quantum state which falls into a black hole sometime after half the black hole has evaporated will be accessible from an observer who knows the dynamics of the black hole and has been collecting all off the Hawking radiation.

\subsection{Parallel circuit model on the complete graph}
Recall that in the parallel circuit model, a random maximum matching of the qubits is chosen and a random gate is applied on each edge of the matching.  Consider figure \ref{fig:decoupling} with the systems $M$ and $M'$ having a constant size $m$ (think of $M'$ as the message), and $A$ and $A'$ are in a maximally entangled state. Clearly, if we have access to the whole output we can recover the message, i.e., a purification of $M$. The following theorem proves that if we apply a parallel random circuit of depth $O(\log n)$, a sufficiently large constant number of randomly chosen qubits of the output together with the system $A$ are sufficient for approximately recovering the message. As mentioned earlier, this is equivalent to proving that the system $M$ is decoupled from a subset of the qubits of size $n-c$ for some constant $c$.

\begin{theorem}
\label{thm:parallel}
Let $\e >0$ and $m$ be a constant. In the setting described above, we have on average over a randomly chosen $T$ of size $n-c$ for some sufficiently large $c$ (depending on $\e$ and $m$) such that
\[
\left\| \rho_{MT} - \frac{\1}{2^m} \otimes \frac{\1}{2^{|T|}} \right\|_1 \leq \e
\]
with probability $1- O(\log^3 n/n)$ over the choice of the circuit. Here, $\rho_{MT}$ refers to the state you obtain by applying a random parallel circuit of depth $O(\log n)$.
\end{theorem}
\begin{proof}
We can write the initial state on $MM'A'$ as
\[
\rho(0) = \frac{1}{2^{2m}} \sum_{\nu \in \{0,1,2,3\}^m} \sigma_{\nu} \otimes \sigma_{\nu} \otimes \frac{\1}{2^{n-m}}.
\]
Then we have, as in the proof of Theorem \ref{thm:strong-scrambling-seq},
\begin{align*}
\ex{\left\| \rho(t)_{MT} - \frac{\1}{2^{m + n-c}} \right\|^2_1} 
&\leq \sum_{ \nu \in \{0,1,2,3\}^m, \mu \in \{0,1,2,3\}^{T}, (\nu, \mu) \neq 0} \ex{\tr[ \sigma_{\nu} \otimes \sigma_{\mu} \otimes \1 \rho(t) ]^2} \\
&= \sum_{ \nu \in \{0,1,2,3\}^m, (\nu, \mu) \neq 0 } \ex{\tr[ \sigma_{\nu} \otimes \sigma_{\mu} \rho(t) ]^2} \\
&= \sum_{ \nu \in \{0,1,2,3\}^m, \nu \neq 0 } \tr[\sigma_{\nu} \otimes \sigma_{\nu} \otimes \1 \rho(0)]^2 \pr{S_t(\supp(\nu)) \subseteq T} \\
&= O(\pr{S_t(\{1\}) \subseteq T})
\end{align*}
where we defined the Markov chain $\{S_t\}$ whose state space is the set of subsets $[n]$ which corresponds to the set of non-zero Pauli operators. The transition probabilities of the Markov chain are defined as follows. We start by choosing  a random maximum matching of the nodes. For each edge $\{i,j\}$ of the matching, we do the following: if neither $i$ nor $j$ are in $S_t$, they are still not in $S_{t+1}$, but if one of the nodes $\{i, j\}$ is in $S_t$, then with probability $9/15$, $i$ and $j$ are in $S_{t+1}$ and with probability $3/15$, $i \in S_{t+1}$ and $j \notin S_{t+1}$ and with probability $3/15$, $j \in S_{t+1}$ and $i \notin S_{t+1}$. As before, we use the notation $S_t(A)$ when the Markov chain starts in the state $A$. Here our Markov chain is assumed to start in the state $\{1\}$, so we will drop the $(\{1\})$ from now on.

\begin{lemma}
\label{lem:mc-subsets}
For a sufficiently large constant $c$ and $t \geq c \log n$ and sufficiently small constant $f$, we have
\[
\pr{|S_t| \leq f n} \leq O\left(\frac{\log^3 n}{n} \right).
\]
\end{lemma}
Before proving the lemma, we just note that it is sufficient to prove the desired result. In fact, for a randomly chosen $T$ of size $n-c$, we have
\begin{align*}
\pr{S_t \subseteq T} &\leq \pr{\forall x \in [n]-T, x \not\in S_t, |S_t| > fn} + \pr{|S_t| \leq f n} \\
				&\leq \pr{x \not\in S_t, |S_t| > fn}^c + O\left(\frac{\log^3 n}{n} \right) \\
				&\leq (1-f)^c + O\left(\frac{\log^3 n}{n} \right)		 
\end{align*}
where $x$ is uniformly distributed on $[n]$. This proves the theorem.
\end{proof}

\begin{proof}[of Lemma \ref{lem:mc-subsets}]
The analysis has two steps. The first part of the proof deals with the case where $S = O(\log n)$ and the second part with the case $S = \Omega(\log n)$.
%We first prove that within $c_1 \log n$ steps of the walk, we will have $|S_t| \geq 10 \log n$ with probability at least $1-\frac{\log^3 n}{n}$. Let us first compute the probability that.

Define $T_1 = \min \{ t : |S_t| \geq 10 \log n\}$. The pre-factor $10$ is chosen only for concreteness and can of course be chosen to be any constant and the statement remains unchanged. We start by proving that 
\[
\pr{T_1 \geq c_1 \log n} = O \left(\frac{\log^3 n}{n} \right).
\]
Let $\eventfont{E}$ be the event that for all $s \leq c_1 \log n$, nodes $i, j \in S_s$ never get matched. We have
\begin{align}
\label{eq:condition-no-collision}
\pr{T_1 \geq c_1 \log n} &= \pr{T_1 \geq c_1 \log n, \eventfont{E}} + \pr{T_1 \geq c_1 \log n, \eventfont{E}^c}.
\end{align} 
Let us analyze the second term first. For this, we denote the matching by $\{(M^1_k, M_k^2)\}_{1 \leq k \leq n/2}$.
\begin{align*}
\pr{T_1 \geq c_1 \log n, \eventfont{E}^c} &= \pr{ \exists s \in [c_1 \log n], k \in [n/2] : M^1_k,M^2_k \in S_s, T_1 \geq c_1 \log n} \\
&\leq \sum_{s = 1}^{c_1 \log n} \pr{ \exists k \in [n/2] : M^1_k,M^2_k \in S_s, S_s \leq 10 \log n} \\
&\leq c_1 \log n \cdot \frac{10 \log n \cdot (10 \log n - 1)}{2n} \\
&\leq \frac{100 c_1 \log^3 n}{n}.
\end{align*}
We now focus on the first term in \eqref{eq:condition-no-collision}. Because $\eventfont{E}$ holds, we know that $|S_{s+1}| \geq |S_s|$ for all $s \in [c_1 \log n]$. More precisely, if $|S_s| = k$, we have $|S_{s+1}|$ is distributed as $k + \bin(k, 9/15)$. But using a Chernoff-Hoeffding bound, we have
\[
\pr{\bin(k, 9/15) \leq k/2} \leq e^{-\frac{1}{2 \cdot 3/5} k (3/5-1/2)^2} \leq e^{-k/200}.
\]
We can now define the times $T(i) = \min \{ t : |S_t| \geq 2^i \}$. With this notation, and letting $m = \log (10 \log n)$, we have $T_1 = T(m)$. We can write
\begin{align*}
\pr{T_1 \geq c_1 \log n, \eventfont{E}}
&\leq \pr{T(1) \geq c'_1 2^{m-1}} + \pr{T(1) < c'_1 2^{m-1}, T(2) \geq c'_1 (2^{m-1} +  2^{m-2}) } + \\
&\qquad \dots +\pr{T(1) + \cdots + T(m-1) < c'_1 (2^{m-1} + \dots 2^{1}), T(m) \geq c'_1 (2^{m-1} + \dots + 1)} \\
&\leq e^{-1/200 \cdot c'_1 2^{m-1}} + e^{-2/200 \cdot c'_1 2^{m-2}} + \dots + e^{-2^{m-1}/200 \cdot c'_1} \\
&= m \cdot e^{-c'_1 2^{m-1}/200}.
\end{align*}
where $c'_1$ is chosen so that $c'_1 (2^{m} - 1) = c_1 \log n$. For $c_1$ (or equivalently $c'_1$) large enough, this expression is at most $1/n$.

For the second part, we consider a large enough subset $S \subseteq [n]$ and we prove that in one step of the random circuit, the size will increase by a constant fraction with high probability. Now it is not possible to assume that we do not have any gates within $S$ itself. But the fact that $S$ is large, we can have better concentration. First given an $S$, let $N_S$ be the number of gates that are between two nodes of $S$. It is easy to see that the expected number of such gates is $\ex{N_S} = \frac{|S|(|S|-1)}{2} \cdot \frac{1}{n-1}$. Actually what we want is to bound $\pr{N_S > \beta |S|}$ where $\beta$ is some small constant to be fixed later. We could use a straight Markov inequality
\[
\pr{N_S > \beta |S|} \leq \frac{\ex{N_S}}{\beta |S|} = \frac{ (|S|-1)}{2\beta (n-1)},
\]
which is good enough for $|S| = o(n)$ but does not give a good bound for linear $|S|$. That's why we compute the second moment of $N_S$.
\begin{align*}
\ex{N_S^2} &= \sum_{i < j, k < l} \ex{\1_{(i,j)} \1_{(k,l)}} \\
			&= \sum_{i < j} \ex{\1_{(i,j)}} + \sum_{i<j, k < l,k, l \notin \{i,j\}} \ex{\1_{(i,j)} \1_{(k,l)}} \\
			&= \frac{|S|(|S|-1)}{2} \frac{1}{n-1} + \frac{|S|(|S|-1)}{2} \cdot \frac{(|S|-2)(|S|-3)}{2} \cdot \frac{1}{(n-1)(n-3)}. 
\end{align*}
where $\1_{(i,j)}$ is one if there is a gate applied between nodes $i$ and $j$, which are both in $S$. Thus the variance is equal to
\begin{align*}
\ex{N_S^2} - \ex{N_S}^2 &= \frac{|S|(|S|-1)}{2} \frac{1}{n-1} + \frac{|S|(|S|-1)}{2} \cdot \frac{(|S|-2)(|S|-3)}{2} \cdot \frac{1}{(n-1)(n-3)} - \frac{|S|^2(|S|-1)^2}{4} \cdot \frac{1}{(n-1)^2} \\
&= \frac{|S|(|S|-1)}{2} \frac{1}{n-1}\left( 1 + \frac{(|S|-2)(|S|-3)}{2} \cdot \frac{1}{n-3} - \frac{|S|(|S|-1)}{2} \frac{1}{n-1} \right).
\end{align*}
But $(|S|-2)(|S|-3) (n-1) = (|S|^2 - 5 |S| + 6) (n-1) = (|S|^2 - |S|)(n-1) - 2(2|S|-3) (n-1)$ and we compare that to $(|S|^2 - |S|) (n - 3) = (|S|^2 - |S|) (n-1) - 2 (|S|(|S|-1))$. The first term is smaller than the second one provided $|S| \geq 3$ and $|S| \leq n$ (which is the case). Thus we can bound the variance by the expected value:
\[
\ex{N_S^2} - \ex{N_S}^2 \leq \ex{N_S}.
\]
By applying Chebyshev's inequality, we have for any $\gamma > 0$,
%\[
%\pr{N_S > \ex{N_S} + \alpha \sqrt{\ex{N_S}} } \leq \frac{1}{\alpha^2}.
%\]
%Setting $\alpha = \gamma \sqrt{\ex{N_S}}$, we get 
\[
\pr{N_S > (1 + \gamma) \ex{N_S} } \leq \frac{1}{\gamma^2 \ex{N_S}}.
\]
Now if $\beta$ is such that $|S| < 2 \beta n$, we define $\gamma$ so that $(1+\gamma) = \beta |S|/\ex{N_S} = 2 \beta \frac{n-1}{|S|-1}$. As a result,
\begin{align*}
\pr{N_S > \beta |S|} &\leq \frac{1}{(\beta \frac{|S|}{\ex{N_S}} - 1 )^2 \ex{N_S}} \\
			&= \frac{2 (n-1)}{(\beta \frac{2(n-1)}{|S|-1} - 1 )^2 |S| (|S|-1)} \\
			&\leq \frac{2 (n-1)}{( 2\beta (n-1) - (|S| - 1) )^2 } \\
			&= O\left(\frac{1}{n} \right),
\end{align*}
provided for example $|S| - 1 \leq  \beta (n-1)$.

We proved that we can assume that the number of gates within $S$ is small. For the gates that associate a node in $S$ with a node outside $S$, we need to prove that many of these gates lead to a Pauli operator of weight two so that we obtain an overall increase in the size of $S$. In fact, provided $N_S < \beta |S|$, the number of non-zero Pauli operator is distributed at least as $(1-\beta) |S| + \bin( (1-2\beta) |S|, 3/5)$. Now we can bound using a standard Chernoff bound
\begin{align*}
\pr{\bin( (1-2\beta) |S|, 3/5) < (\beta + 1/4) \cdot |S|} 
&\leq \exp{-\frac{5}{6} \cdot \frac{\left(3/5 (1-2\beta) |S| - (\beta + 1/4) |S|  \right)^2}{(1-2\beta)|S|}} \\ 
&=\exp{-\frac{5 |S| }{6} \cdot \frac{\left(7/20 - (6/5 + 1)\beta\right)^2}{(1-2\beta)}}. \\
\end{align*} 
For sufficiently small $\beta$ and sufficiently large $|S| = \Omega(\log n)$, this probability is $O(1/n)$.
%Just for concreteness: I am using the bound for $q < p$: $\pr{Bin(n,p) \leq k} \leq e^{-\frac{1}{2p} \frac{(np-k)^2}{n} }$.

This proves that provided $c \log n \leq |S_s| \leq \beta n$ for a sufficiently large $c$ and sufficiently small $\beta$, then we have $|S_{s+1}| \geq 5/4 \cdot |S_{s}|$. Together with the first part of the proof, we obtain that after $O(\log n)$ steps of the random circuit we have $|S_t| \geq \beta n$ with probability $1 - O(\log^3 n/n)$. 
 % the weight of the Pauli operator increases the weight by a factor of $5/4$ and so after logarithmically many steps we achieved constant fraction.
\end{proof}

%\subsection{Parallel circuit model on a $d$-dimensional lattice}
\subsection{Random circuit on a $d$-dimensional lattice}
We now turn to examining the depth of a random quantum circuit restricted to nearest neighbours on a $d$-dimensional square lattice, required to scramble a constant number of initial low weight Pauli strings.  As was shown in the previous section, this determines the depth at which we obtain entanglement assisted codes.
%capacity for the quantum erasure channel can be achieved.  

%This quantity as developed in \cite{BlackHoles} can be used as an estimate for the time required to retrieve a quantum state dropped into a black hole for collecting the subsequent Hawking radiation. 

 We analyze scrambling in a model for which a partial parallelization has been performed of a sequential random quantum circuit on on a $d$-dimensional lattice.  By a $d$-dimensional sequential random quantum circuit we mean one for which a random two-qubit gate selected according to the Haar measure or uniformly from the Clifford group is applied to a pair of qubits selected uniformly from among nearest neighbors on a square $d$-dimensional lattice with open boundary conditions. The specific model consists of partitioning the lattice into coarse grained cells and in each time step performing a  random two-qubit gate between a randomly selected pair of nearest neighbour within each cell.  We consider two equivalent coarse grainings of the lattice into square cells of size $O(\log n)$, such that the midpoints of the cells of the first set are the corners of the cells of the second set. We then have cells of type $1$ corresponding to the first coarse graining and cells of type $2$ corresponding to the second coarse graining. In alternating coarse time steps,  gates are applied within each of the the cells of one set at a time.  In each coarse time step a total of  $O(\log^2 n)$ gates will be performed.
 
The following theorem proves that after $O((\frac{n}{\log n})^{1/d})$ coarse grained time steps, a sufficiently large constant number of randomly chosen qubits of the output together with the system $A$ are sufficient for approximately recovering the message (see Figure \ref{fig:decoupling}). As mentioned earlier, this is equivalent to proving that the system $M$ is decoupled from a subset of the qubits of size $n-c$ for some constant $c$.

%Recall that in the parallel circuit model, a random maximum matching of the qubits is chosen and a random gate is applied on each edge of the matching. For this model, we are only able to prove a weaker scrambling result. Consider the figure \ref{fig:decoupling} with the systems $M$ and $M'$ having a constant size $m$ (think of $M'$ as the message), and $A$ and $A'$ are in a maximally entangled state. Clearly, if we have access to the whole output we can recover the message, i.e., a purification of $M$. The following theorem proves that if we apply a parallel random circuit of depth $O(\log n)$, a sufficiently large constant number of randomly chosen qubits of the output together with the system $A$ are sufficient for approximately recovering the message. As mentioned earlier, this is equivalent to proving that the system $M$ is decoupled from a subset of the qubits of size $n-c$ for some constant $c$.

\begin{theorem}
\label{ddim}
Let $\e >0$ and $m$ be a constant. In the setting described above, we have on average over a randomly chosen $T$ of size $n-c$ for some sufficiently large $c$ (depending on $\e$) such that
\[
\left\| \rho_{MT} - \frac{\1}{2^m} \otimes \frac{\1}{2^{|T|}} \right\|_1 \leq \e
\]
with probability $1- O(1/n)$ over the choice of the circuit. Here, $\rho_{MT}$ refers to the state you obtain by applying a random quantum circuit of depth $O(n^{1/d} \log^2(n))$ as described above.
\end{theorem}

\begin{proof}
As in the proof of Theorem \ref{thm:parallel}, we only need to show that a Pauli string of weight one becomes a Pauli string of linear weight within $O(n^{1/d} \log^2 n)$ time steps with high probability. 
%The proof proceeds as follows. We show that there is a large probability that after a coarse step applied on a cell that has at least one non-identity Pauli string, the resulting Pauli string in the cell has support in all the regions overlapping with a neighbouring cells of the next coarse time step.  
We start by proving (in Theorem \ref{thm:gaps} below) a lower bound of $\Omega(1/n)$ on the gap of the second moment operator of a sequential random quantum circuit on $n$ qubits with a $d$-dimensional lattice interaction graph (or equivalently on the corresponding the Markov chain described in Section \ref{sec:prelim-rqc}). Then by a standard argument, one can obtain an upper bound on the mixing time of the Markov chain; see e.g., \cite{MT06}. That is after $t = O(\frac{1}{\Delta} (n + \log(1/\delta)))$, the distribution on Pauli strings is $\delta$-close to the the stationary distribution of the Markov chain which is the uniform distribution over all non-zero Pauli strings.

We then apply this result to a cell which is a $d$-dimensional lattice with $O(\log n)$ nodes.  Then partition the cell into $2^d$ $d$-dimensional sub-cells with half the length of the original cell. Then, if we choose $\delta$ to be inverse polynomial with a sufficiently large power, and applying $O(\log n ( \log n + \log(1/\delta))) = O(\log^2n)$ gates, a non-zero Pauli string in the cell gets mapped to a Pauli string whose support contains at least on element in each one of these sub-cells with probability $1-1/\poly(n)$. We can summarize this as follows: in each successful coarse time step every cell that contains a non-zero Pauli string gets mapped to a Pauli string that has support in each one of its sub-cells. By choosing the constants appropriately, we can make the probability of success of a coarse step to be $1-1/\poly(n)$.

Consider now the following coarse time step that uses the alternate coarse graining.
What the success of the previous coarse step is saying is that a cell with non-zero Pauli weight has contaminated all the cells of type $2$ that overlap with it. By the same argument as in the previous paragraph, we see that in the coarse step involving cells of type $2$, each one of these contaminated cells of type $2$ will in turn contaminate the cells of type $1$ that overlap with it. Thus, by repeating these alternate steps a number of times corresponding to the diameter of the graph of cells $O\left( \left(\frac{n}{\log n}\right)^{1/d}\right)$, we reach a Pauli string of linear weight.
%Note that each one of these sub-cells is a part of a different cell in the second coarse gra
%we obtain that after applying $O(\log^2 n)$ randomly chosen gates within the cell.
%Then by taking cells of size $O(\log n)$, one can prove that each overlap region contains the support of a Pauli string with probability  $1-1/\poly(n)$.  After repeating these alternating coarse time steps $O((\frac{n}{\log n})^{1/d})$ times (which is the diameter of the graph where each node is a coarse grained cell), we obtain a Pauli string with a linear weight. We skip this part of the proof (which simply involves relating the spectral gap to the mixing time and some applications of Chernoff-type bounds) for brievity and concentrate on bounding the gap of the second moment operator for $d$-dimensional random quantum circuits.
\end{proof}

\begin{theorem}
\label{thm:gaps}
The spectral gap, $\Delta_{ds}$, of the second-order moment operator for the $d$-dimensional sequential random quantum circuit  described above is bounded from below by $\Delta_{ds}\ge \frac{a}{n} $ for a constant $a$.
\end{theorem}
\begin{proof}
The second-order moment operator of a sequential one-dimensional random quantum circuit is of the form, $$M_{1s}=\frac{1}{n-1}\sum_{i=1}^{n-1} m_{ii+1}.$$  We will use the fact that the gap, $\Delta_{1s}$, of the second-order moment operator of a sequential $1D$ random quantum circuit  was shown in  \cite{Znidaric2, BrandaoRH} to be lower bounded by $\Delta_{1s}\ge \frac{a}{n}$, to show a similar lower bound on the spectral gap, $\Delta_q$, of the second-order moment operator for a non-uniform sequential $1D$ random quantum circuit, for which the probability of applying a gate to qubits $i$ and $i+1$ is $\frac{q_{ii+1}}{n-1}$ with $q_{ii+1} > 0$. We next show how to write the second-order moment operator of the $d$-dimensional sequential random quantum circuit as a convex sum of such non-uniform $1D$ sequential random quantum circuits, which we will lead to the desired bound by using the following lemma on the convexity of the spectral gap.

\begin{lemma}
\label{thm:adding}
For two random quantum circuits whose gate distributions are universal and are invariant under Hermitian conjugation,  with second-order moment operators $M_1$ and $M_2$ with spectral gaps of $\Delta_1$ and $\Delta_2$ respectively, the second-order moment operator describing any convex combination of the two random quantum circuits, $M=p_1M_1 + p_2M_2$, has a gap $\Delta$ that is lower bounded by  $\Delta \ge p_1\Delta_1 + p_2\Delta_2$.  %For two random quantum circuits whose gate distributions are universal for which the corresponding moment operators are not Hermitian it follows that, $s \le p_1s_1+ p_2s_2$, where $s$, $s_1$, and $s_2$ are the second largest singular values of $M$, $M_1$, and $M_2$.
\end{lemma}
\begin{proof}
We use the fact that every second-order moment operator of a random quantum circuit over a universal gate set has the same space of fixed points, $\mathcal{V}$, onto which the second-order moment operator of the Haar measure, $M_{\haar}$, is the projector.   Consequently, we may define the following operators $\tilde{M}=M-M_{\haar},\tilde{M_1}=M_1-M_{\haar}, \tilde{M_2}=M_2-M_{\haar}$ which sets the eigenvalue of this eigenspace to 0.  Now by the triangle inequality it follows that, $$\|\tilde{M}\|_{sp} \le  p_1\|\tilde{M_1}\|_{sp} +p_2\|\tilde{M_2}\|_{sp},$$
 where $\|~\|_{sp}$ is the spectral norm.  Since invariance under Hermitian conjugation of the gate distribution implies that the moment operators are Hermitian, it follows  that $\lambda \le p_1\lambda_1 +p_2\lambda_2$, where $\lambda$, $\lambda_1$ and $\lambda_2$ are the subdominant eigenvalues $M$, $M_1$ and $M_2$, respectively.  Since $\Delta = 1-\lambda$, the lemma follows. % If the moment operators are not Hermitian, the triangle inequality applies instead to the second largest singular values of $M$, $M_1$, and $M_2$.
\end{proof}

The second-order moment operator for a non-uniform $1D$ sequential random quantum circuit is given by,
$$M_{q}=\frac{1}{n-1}\sum_{i=1}^{n-1} q_{ii+1}m_{ii+1}.$$
Observe that because $m_{ii+1}$ are positive semidefinite, we have
\[
M_{q} \geq \min_i q_{ii+1} \cdot  \frac{1}{n-1} \sum_{i=1}^{n-1} m_{ii+1}.
\]
This implies that the spectral gaps satisfy $\Delta_{q} \geq \min_i q_{ii+1} \Delta_{1s} = \min_i q_{ii+1} \cdot \Omega(1/n)$.
%Let $\ket{\psi_q}$ be the vector orthogonal to  $\mathcal{V}$ such that  $ \sum_i q_{ii+1}\bra{\psi_q}m_{ii+1}\ket{\psi_q}$ is a minimum.  This value is equal to the spctral gap,  $\Delta_q$ of $M_{1s}$.  Using the fact that each $m_{ii+1}$ is positive semidefinite, it follows that, 
%$\Delta_q \ge \min(q_{ii+1}) \sum_i \bra{\psi_q}m_{ii+1}\ket{\psi_q}$.  But since $\bra{\psi_q}m_{ii+1}\ket{\psi_q} \ge  \min_{\psi\perp \mathcal{V}} \sum_i \bra{\psi}m_{ii+1}\ket{\psi} = \Delta_{1d}$, it follows that $\Delta_q \ge \min(q_{ii+1}) \Delta_{1s}$.

\begin{figure}[h]
	\begin{center}
		\includegraphics[scale=0.8]{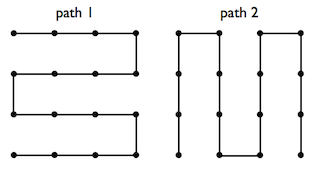}
		\caption{A set of non-intersecting 1D-paths on a 2D 4x4 square lattice}
		\label{fig:paths}
\end{center}
\end{figure}

The goal now is to write the second moment operator for the $d$-dimensional circuit as a convex combination of second moment operators for one-dimensional circuits. For this, we find a set of paths on the $d$-dimensional lattice such that each path includes every vertex and each edge is included in at least one path.   Such a set may be constructed using $d$ paths where the $i$-th path consists of every edge oriented in the $i$-th direction plus some perpendicular edges on the surface of the lattice.   An illustration of such a set is given in Figure \ref{fig:paths} for $d=2$. For such a set of paths every internal edge is traversed by only one path, while an external edge may be traversed by as many as $d$ paths.  Thus, one may write the $d$-dimensional sequential random quantum circuit as an average of $d$ one-dimensional non-uniform sequential random quantum circuits where the pair with the lowest probability is $1/d$ of that of the largest.  Lemma \ref{thm:adding} now implies that the gap, $\Delta_{ds}$, of the $d$-dimensional sequential random quantum circuit is bounded by $\Delta_{ds} \ge \frac{1}{d}\Delta_{1s} \ge \frac{a}{n}$ for a constant $a$.    
\end{proof}

\subsection{Lower bound on the scrambling time}
For a circuit of depth $t$ consisting of gates, each of which act on at most $k$-qubits, an initial Pauli operator of weight 1 can only have support on a qubit distance $kt$ away.  Thus, on a $d$-dimensional lattice it may have support on at most $(kt)^d$ qubits, implying a depth of at least $\Omega(n^{1/d})$ for the weigh to be linear in $n$.  For a random quantum circuit on the complete graph, the weight of a Pauli operator may increase by at most a factor of $k$ in each time step, yielding an lower bound on the depth required for scrambling of $\Omega(\log n)$.   Thus, the time at which most random quantum circuit scramble is within a constant or $O(\log^2 n)$ factor of the fastest possible circuit on the complete graph and a $d$-dimensional graph respectively.  We think that  the $O(\log^2 n)$ factor in the case of the $d$-dimensional lattice is an artifact of our proof technique.

\section{Outlook}

An interesting question is whether decoupling occurs in the general setting of \cite{Dup09,DBWR10} with random circuits of depth $O(\poly(\log n))$.   This would imply that random encoding circuits of $O(\poly(\log n))$ depth generate codes that are close to achieving the capacity of the erasure channel.   

Since the scrambling condition in $d$-dimensional random quantum circuits utilizes a weak bound on the success probability of filling a sufficient number of cells, we think it may be possible to tighten our result to show strong scrambling, and thus stronger decoupling results for circuits of depth $O(n^{1/d} \log n)$ on $d$-dimensional lattices.  It would be interesting to see if the course graining technique employed here can be used to show that the depth at which random quantum circuits are $\epsilon$-approximate $k$-designs also scales with the radius of the interaction graph as conjectured in \cite{BHH12}. 

It is known that the unitary generated by an arbitrary local Hamiltonian at time, $t$, which is a polynomial in the size, $n$, of the system can be approximated by a circuit consisting of single and two-qubit gates whose depth is  polynomial in $t$ \cite{illusion}.  Thus, whether our results imply that Hamiltonian evolution scrambles quickly depends on whether typical time independent Hamiltonians explore sufficiently uniformly the measure accessible to them at times sublinear in $n$.  This question appears to be linked with the approach of random matrix theory \cite{SrednickiETH, RigolETH} to understand thermalization under dynamics generated by strongly-nonintegrable, time independent Hamiltonians, whereby the eigenstates of the Hamiltonian resemble those drawn uniformly from an appropriate matrix ensemble.  It would be interesting to further explore the connection between quantum chaos, properties of random quantum circuits and quantum aspects of thermalization such as scrambling and decoupling.

\iffalse 
It would be interesting to explore the connection between scrambling and thermalization. In \cite{SrednickiETH} it was proposed that for closed quantum many body systems, the approach to thermal equilibrium takes place  at the level of the eigenstates. For non-integrable quantum systems, it was conjectured in \cite{RigolETH}, that the eigenstates can be described statistically, and that moreover, this statistical structure maps itself onto time evolved states. The time scale at which this occurs for a large isolated system is the question that we have addressed partially in the context of quantum circuits.    It is known that the unitary generated by a an arbitrary local Hamiltonian at time polynomial in the size of the system can be approximated by a circuit consisting of singe and two qubit gates whose depth is linear in $t$ \cite{illusion}.  Thus, whether our results imply that Hamiltonian evolution scrambles quickly depends on whether typical time independent Hamiltonians explore sufficiently uniformly the measure accessible to them at times sublinear in $n$.  Numerical evidence showing that eigenvectors and eigenvalues of strongly non-integrable many-body Hamiltonians resemble those of random matrix ensembles suggests that this is the case at least for time scales that are exponential in system size.  Our results raises questions in the mathematical theory of spectral fluctuations because of the need to answer whether dynamical features of non-integrable local Hamiltonians display concentration of measure at sublinear time scales.

\fi

\section*{Acknowledgements}

We would like to thank Patrick Hayden and David Poulin for helpful discussions. The research of WB is supported by the Centre de Recherches Math\'ematiques at the University of Montreal, Mprime, and the Lockheed Martin Corporation. The research of OF is supported by the European Research Council grant No. 258932, and was started while he was affiliated with McGill University.

\appendix

\section{Analysis of the Markov chain}
\begin{theorem}
\label{thm:mc-convergence}
Let $X_t(\ell)$ be the random variable representing the position of the random walk starting at $\ell$ after $t$ steps. There is a constant $c$ such that for any $f < 1/2$ and $t \geq c n \log^2 n$ and all $\ell \in \{1,\dots,n\}$,
%\begin{equation}
%\label{eq:binomial-avg}
%\frac{1}{2^n - 1} \sum_{\ell=1}^n \binom{n}{\ell} \pr{X_t(\ell) \leq fn} \leq \e 2^n
%\end{equation}
%A better statement would be for all $\ell$, 
\[
\pr{X_t(\ell) \leq f n} \leq \frac{2^{f \log 3 + h(f)}}{\binom{n}{n/2} 3^{n/2}} + \frac{1}{2^{\ell} \binom{n}{\ell}}\frac{1}{\poly(n)}.
\]
\end{theorem}
\comment{I conjecture that one should be able to obtain a bound of $\frac{2^{f \log 3 + h(f)}}{4^n} + \frac{1}{3^{\ell} \binom{n}{\ell}} \frac{1}{\poly(n)}$ but I am not sure how exactly. For the first term, we should take instead of $n/2$ as a reference point the point $3n/4$ (which caries a significant fraction of the total probability mass of the stationary distribution) or something close to it, but the problem then is as $r$ gets closer to $3n/4$, the second term gets worse. I think I can get such a bound for $\ell \leq \sqrt{n}$ or something}
\begin{proof}
The general strategy of the proof is as follows. First, we pick a reference point $r$ (which is going to be $n/2$) for which we can bound $\pr{X_t(r) \leq fn}$ easily. Then we will prove that for any $\ell < r$, starting at $\ell$, we will reach $r$ within $t$ steps with high probability.

The stationary distribution of the chain is given by $\pi(k) = \frac{3^k \binom{n}{k} }{4^n - 1}$ (see \cite[Lemma 5.3]{HL09}). As a result, we have for any $t \geq 1$, 
\begin{align*}
\frac{1}{4^n-1} \sum_{\ell=1}^n 3^{\ell} \binom{n}{\ell} \pr{X_t(\ell) \leq f n}
&= \frac{1}{4^n-1} \sum_{\ell=1}^{fn} 3^{\ell} \binom{n}{\ell} \\
&\leq \frac{3^{fn} 2^{h(f)n}}{4^n-1} \\
&= \frac{2^{ (f \log 3 + h(f))n}}{4^n-1}. 
\end{align*}
%
%To analyse the expression in \eqref{eq:binomial-avg}, we split the sum into two terms $\ell \geq \ell_0$ and $\ell < \ell_0$. The sum $\ell \geq \ell_0$ is easily handled and actually doesn't depend on $t$ being large (it could be $1$):
%\begin{align}
%\frac{1}{2^n-1} \sum_{\ell=\ell_0}^n \binom{n}{\ell} \pr{X_t(\ell) \leq f n} &= \frac{1}{4^n-1} \sum_{\ell=\ell_0}^n \frac{2^n+1}{3^{\ell}} 3^{\ell}\binom{n}{\ell} \pr{X_t(\ell) \leq f n} \notag \\
%&\leq \frac{2^n+1}{3^{\ell_0}} \frac{1}{4^n-1} \sum_{\ell=k}^n 3^{\ell} \binom{n}{\ell} \pr{X_t(\ell) \leq f n} \notag \\
%&\leq \frac{2^n+1}{3^{\ell_0}} \cdot \frac{2^{ (f \log 3 + h(f))n}}{4^n-1}. \label{eq:bound-larger-k}
%\end{align}
This allows us to bound the probability of the event $\event{X_t(n/2) \leq fn}$. In fact, for any $t$,
%Before analysing the sum over $\ell < \ell_0$, we make an observation on $\pr{X_t(n/2) \leq fn}$.
%Using \eqref{eq:bound-larger-k}, we have that
\begin{align}
\label{eq:returnfn-from-middle}
\pr{X_t(n/2) \leq f n} &= \frac{4^n-1}{\binom{n}{n/2} 3^{n/2}} \cdot \frac{\binom{n}{n/2} 3^{n/2}}{4^n-1} \pr{X_t(n/2) \leq f n} \\
&\leq \frac{4^n-1}{\binom{n}{n/2} 3^{n/2}} \cdot 
\frac{1}{4^n-1} \sum_{\ell=1}^n 3^{\ell} \binom{n}{\ell} \pr{X_t(\ell) \leq f n} \notag \\
&\leq \frac{2^{f \log 3 + h(f)}}{\binom{n}{n/2} 3^{n/2}}. \notag
\end{align}

%Provided $f < 1/10$, we will have $\pr{X_t(n/2) \leq f n} \leq \frac{1}{\binom{n}{n/2}} \cdot 2^{-0.1n}$. Thus, for sufficiently large $n$, $\pr{X_t(n/2) \leq f n} \leq  \frac{2^{-n}}{\poly(n)}$.

Moreover, note that for $\ell \geq n/2$, we have 
\begin{align*}
\pr{X_t(\ell) \leq fn} 
&\leq \max_{1 \leq s \leq t} \pr{X_s(n/2) \leq fn} \\
&\leq \frac{2^{f \log 3 + h(f)}}{\binom{n}{n/2} 3^{n/2}}
%\frac{2^{-n}}{\poly(n)}.
\end{align*}

The remaining case is then $\ell \leq n/2$. In this case, the objective is to show that
\[
\pr{X_t(\ell) \leq f n} \leq \frac{2^{f \log 3 + h(f)}}{\binom{n}{n/2} 3^{n/2}} + 
\frac{1}{2^{\ell} \binom{n}{\ell}} \frac{1}{\poly(n)}. 
\]
Define $T = \min \{t \geq 1 : X_t(\ell) \geq n/2 \}$. Note that we have for any $t$
\begin{align*}
\pr{X_t(\ell) \leq f n} &\leq \pr{T < t, X_t(\ell) \leq fn} + \pr{T \geq t} \\
				&= \pr{T < t, X_{t-T}(n/2) \leq fn} + \pr{T \geq t} \\
			&\leq  \max_{1 \leq s \leq t} \pr{X_s(n/2) \leq fn} + \pr{T \geq t}.
\end{align*}
Using \eqref{eq:returnfn-from-middle}, we can bound the first term. The objective of the remainder of the proof is to bound the probability $\pr{T \geq t}$ when $t = c n \log^2 n$. This is done in Lemma \ref{lem:time-reach-middle} below.
Once we have that, the result follows.
\end{proof}
 
\begin{lemma}
\label{lem:time-reach-middle}
For a large enough constant $c$,
\[
\pr{T > c n \log^2 n} \leq 2^{-2n} + \frac{1}{2^{\ell} \binom{n}{\ell} } \cdot \frac{1}{\poly(n)}.
\] 
\end{lemma}
\begin{proof}
To prove this result, we start by defining an accelerated walk $\{Y_i\}$ as in \cite{HL09} and the corresponding stopping time $S = \min \{s : Y_s \geq n/2\}$. More formally, let $N_0 = 0$ and $N_{i+1} = \min \{ k \geq N_i : X_k \neq X_{N_i}\}$ and then $Y_{i} = X_{N_i}$. It is not hard to see that $\{Y_i\}$ is a Markov chain and the transition probabilities are given by the transition probabilities for $\{X_k\}$ conditioned on moving.

We also define the waiting time $W_i = N_{i+1} - N_{i} - 1$ to be the number of times the self-loop edge is taken. Conditioned on $Y_i$, $W_i$ has a geometric distribution with parameter $\frac{2Y_i(3n-2Y_i-1)}{5 n (n-1)}$. Notice that this distribution is stochastically dominated by a geometric distribution with parameter $\frac{2Y_i}{5n}$, which we will use instead (we are only interested in upper bounds on the waiting times).

Getting back to $T$, notice that $T = S+W_1 + W_2 + \dots + W_S$. So we have for all $s$
\begin{equation}
\label{eq:t-s}
\pr{T > t + s} \leq \pr{S > s} + \pr{S \leq s, W_1 + \dots + W_S > t}.
\end{equation}
We will choose $s$ later so that both terms are small.

\begin{lemma}
\label{lem:bound-prob-s}
For any $s \geq 2n$, we have
\[
\pr{S > s} \leq \exp{-s/8}.
\]
\end{lemma}
\begin{proof}
For this we just use a concentration bound on the position of a random walk relative to its expectation. First we define a random walk $Y'_i$ with $Y'_0 = 0$ and it moves to the right with probability $3/4$ and to the left with probability $1/4$. Observe that the probability of moving right is at most $3/4$ for $Y_i$ provided $Y_i \leq n/2$. For this reason, before $S$, we can assume that $Y'_i \leq Y_i$. In other words, we have $S' \geq S$ where $S' = \min \{i : Y'_i \geq n/2\}$.
Thus,
\begin{align*}
\pr{S > s} &\leq \pr{S' > s} \\
			&\leq \pr{Y'_{s} < n/2} \\
			&= \pr{Y'_{s} < \ell + s/2 - (s/2+\ell-n/2)} \\
			&\leq \exp{-\frac{(s/2+\ell-n)^2}{2s} } \\
			&\leq \exp{-s/8}
\end{align*}
where we used the fact that $\ex{Y'_s} = \ell + s/2$ and a Chernoff-type bound, see for example \cite[Lemma A.4]{HL09}.
\end{proof}

We now move to the second step of the proof where we analyze the waiting times $W_1 + \dots + W_S$. Recall this is the total waiting time before the node $r = n/2$ is reached.
\begin{lemma}
\label{lem:bound-waiting-time}
We have
\[
\pr{S \leq s, W_1 + \dots + W_S > c n \log^2 n} \leq \frac{1}{2^{\ell}\binom{n}{\ell}} \cdot\frac{1}{\poly(n)} 
\]
\end{lemma}
\comment{What we really want here is to get a bound of $3^{-\ell} \binom{n}{\ell}^{-1}$ and for the probability of going all the way to 3n/4, but I don't know how to do that except when $\ell$ is small say $\sqrt{n}$.}
\begin{proof}
The techniques we use are similar to the techniques in \cite{HL09}, but we need to improve the analysis in several places. We try to use similar notation as \cite{HL09} as much as possible.

As in the proof of \cite[Lemma A.11]{HL09}, we start by defining the good event 
\[
\eventfont{H} = \bigcap_{x=1}^n \event{\sum_{k=1}^S \1(Y_k \leq x) \leq \gamma x/\mu},
\]
where $\mu = 1/2$.\footnote{We use this notation to apply \cite[Lemma A.5]{HL09} later.  $\mu$ corresponds to the probability of going forward minus the probability of going backward for a simplified walk that moves forward at most as fast as $Y_k$. In our case, we have $\mu = 1/2$ because we stop after reaching state $r=n/2$, and the probability of moving forward at $n/2$ is $3/4$.} The parameter $\gamma$ is going to be chosen later. This event is saying that states with small labels are not visited too many times. \comment{Why 1/2? Actually we don't choose $\mu$, it is the difference between the probabilities of moving forward and backward.} Later in the proof, we will show that the $\pr{\eventfont{H}^c}$ is small.
%Using the same proof as in \cite{HL09} and using the fact that the walk starts at $\ell$, we can show that
%\[
%\pr{\eventfont{H}^c} \leq n \exp{-\frac{\mu \ell (\gamma-2) }{2}}.
%\]
%\comment{Should probably prove this in a later version.}
%
% This is simply wrong just consider x = 1, the probability is at least something like
% \exp( - (\ell + \gamma)) there is no product \ell \cdot \gamma
%
Define the random variable $M = \min_{1 \leq i \leq S} Y_i$. We have
\begin{align}
\pr{W_1 + \dots W_S > t, S \leq s, \eventfont{H}} &= \sum_{m=1}^{\ell} \pr{M=m, S \leq s, W_1 + \dots W_S > t, \eventfont{H}} \notag \\
				&= \sum_{m=1}^{\ell} \pr{M=m} \pr{S \leq s, W_1 + \dots W_S > t, \eventfont{H} | M=m} \notag \\
				%&\leq \sum_{m=1}^{\ell} \pr{M=m} \pr{S \leq s, W_1 + \dots W_S > t ,\eventfont{H}_m | M=m} + \pr{\eventfont{H}_m|M=m} \right) \label{eq:introduce-hm}.
				&\leq \sum_{m=1}^{\ell} \pr{M \leq m}\max_{ \{y_i\} \text{ satisfying } M=m \text{ and } \eventfont{H} \text{ and } S \leq s} \pr{ W(y_1) + \dots + W(y_s) \geq t }, \label{eq:decomp-min}
\end{align}
where the maximum is taken over all sequences $y_1, \dots, y_s$ of possible walks and $W(y)$ is the waiting time at state $y$.

We will bound $\pr{M \leq m}$ using Lemma \ref{lem:hitting-prob}. Our random walk starts at position $\ell$ so that, in the notation of Lemma \ref{lem:hitting-prob}, $p_- = \frac{6 \ell (n - \ell)}{6 \ell (n-\ell) + 2 \ell(\ell-1)}$ and for $k \geq \ell+1$, $p_+(k) = \frac{6 k (n - k)}{6 k (n-k) + 2 (\ell + 1)\ell}$. As a result, we have
\[
\alpha_{-} = \frac{6 \ell(n-\ell)}{2\ell(\ell-1)} = 3 \cdot \frac{n-\ell}{\ell - 1}.
\]

As we stop after reaching the reference point $r = n/2$, we can bound $p_+ \geq 3/4$. As a result, we have
\begin{align*}
\pr{M \leq \ell - 1} &\leq \frac{1}{1 + 3 \cdot \frac{n-\ell}{\ell - 1} \left(1 - 1/3\right)} 
\\
&= \frac{1}{1 + 2 \cdot \frac{n-\ell}{\ell - 1}} \\
&\leq \frac{1}{2} \cdot \frac{\ell - 1}{n - \ell}.
\end{align*}
Reaching $\ell-2$ before $r$ means reaching $\ell-1$ before $r$ starting at $\ell$ and reaching $\ell-2$ before $r$ starting at $\ell-1$, and these parts of the walk are independent. As a result, by induction, we can then see that 
\begin{align}
\pr{M \leq m} &\leq \frac{1}{2^{\ell - m}} \cdot \frac{(\ell-1) (\ell - 2) \cdots m}{(n-\ell) (n-\ell+1) \cdots (n-m-1)} \notag \\
		&= \frac{1}{2^{\ell}} \cdot \frac{\ell !}{n (n-1) \cdots (n-\ell+1)} \cdot \frac{2^m}{\ell (n-\ell)} \cdot \frac{n (n-1) \cdot (n-m)}{(m-1)!} \notag \\
		&\leq \frac{1}{2^{\ell} \binom{n}{\ell}} \cdot (2n)^m. \label{eq:bound-prob-m}
\end{align}

%We have $\pr{M \leq m} \leq \frac{p_+}{2p_+ - 1} \cdot \left(\frac{1-p_{-}}{p_{-}}\right)^{\ell-m}$ with $1 - p_{-} = \frac{2 \ell (\ell -1)}{6 \ell (n-\ell) + 2 \ell(\ell-1)} \leq \ell/(2n)$ as $\ell \leq n/2$ and $p_+ = \frac{6 n^2/4 }{6 n^2/4 + 2 n/2 (n/2-1)} \geq 3/4$. This shows that $\pr{M \leq m} \leq \left(\frac{\ell}{n}\right)^{\ell-m}$.

We now look at the term $\max_{ \{y_i\} \text{ satisfying } M=m \text{ and } \eventfont{H} \text{ and } S \leq s} \pr{ W(y_1) + \dots + W(y_s) \geq t }$. As argued in the proof of \cite[Lemma A.11]{HL09}, the maximum is achieved when we make the walk visit as many times as possible the states with smaller labels. This means state $m$ is visited $\gamma m/\mu$ times, and all $i > m$ are visited $\gamma/\mu$ times. So we can write 
\[
W(y_1) + \dots + W(y_s) \leq \sum_{i=1}^{\gamma m/\mu} G_{m,i} + \sum_{i=1}^{\gamma/\mu} \sum_{k=m+1}^{n/2} G_{k, i},
\]
where $G_{k,i}$ has a geometric distribution with parameter $2k/5n$ and the random variables $\{G_{k,i}\}$ are independent.
We are going to give upper tail bounds on the right hand side by computing the moment generating function. For any $\lambda \geq 0$, we have, using the moment generating function of a geometric distribution and the independence of the random variables:
\begin{align*}
\ex{\exp{\lambda\left(\sum_{i=1}^{\gamma m/\mu} G_{m, i} + \sum_{i=1}^{\gamma/\mu} \sum_{k=m+1}^{n/2} G_{k, i}\right)}} &= \left( \frac{ 2m/5n }{ e^{-\lambda} - 1 + 2m/5n} \right)^{\gamma m / 2} \prod^{n/2}_{k=m+1} \left( \frac{2k/5n}{e^{-\lambda} - 1 +2k/5n} \right)^{\gamma/\mu}.
\end{align*}
Now take $\lambda$ so that $e^{\lambda} = \frac{1}{1-m/(5n)}$. This leads to
\begin{align*}
\ex{\exp{\lambda\left(\sum_{i=1}^{\gamma m/\mu} G_{m, i} + \sum_{i=1}^{\gamma/\mu} \sum_{k=m+1}^{n/2} G_{k, i}\right)}} 
&= \left(\frac{2m}{2m - m}\right)^{\gamma m/\mu} \cdot \prod^{n/2}_{k=m+1} \left( \frac{2k}{2k - m} \right)^{\gamma/\mu} \\
&\leq 2^{\gamma m/\mu} \left(\prod_{k=m+1}^{n/2} e^{\frac{m/2}{k-m/2}} \right)^{\gamma/\mu} \\
&\leq 2^{\gamma m/\mu} \left(e^{m/2 \cdot \ln n} \right)^{\gamma/\mu}.
\end{align*}
As a result, using Markov's inequality, we obtain
\begin{align*}
\pr{\sum_{i=1}^{\gamma m/\mu} G_{m,i} + \sum_{i=1}^{\gamma/\mu} \sum_{k=m+1}^{n/2} G_{k, i} > t}
&= \pr{\exp{\lambda\left( \sum_{i=1}^{\gamma m/\mu} G_{m,i} + \sum_{i=1}^{\gamma/\mu} \sum_{k=m+1}^{n/2} G_{k, i} \right) } > e^{\lambda t}}  \\
&\leq 2^{\gamma m/\mu} e^{\gamma m/(2\mu) \cdot \ln n} \cdot (1-m/5n)^{t} \\
&\leq 2^{\gamma m/\mu} e^{\gamma m/(2\mu) \cdot \ln n} \cdot e^{-tm/(5n)}.
\end{align*}
Getting back to equation \eqref{eq:decomp-min}, we have
\begin{align*}
\pr{W_1 + \dots W_S > t, S \leq s, \eventfont{H}} 
&\leq \frac{1}{2^{\ell} \binom{n}{\ell} } \sum_{m=1}^{\ell} (2n)^m  2^{\gamma m/\mu} e^{\gamma m/(2\mu) \cdot \ln n} \cdot e^{-tm/(5n)}  \\
&= \frac{1}{2^{\ell} \binom{n}{\ell} } \sum_{m=1}^{\ell} \left( 2n 2^{\gamma/\mu} e^{\gamma/(2\mu) \cdot \ln n} \cdot e^{-t/(5n)} \right)^m
\end{align*}
Thus, for $t > c n \log^2 n$ with sufficiently large $c$, this probability is bounded by $O\left(\frac{2^{-\ell}}{\binom{n}{\ell}\poly(n)}\right)$. \comment{Actually, we could even do $2^{c \log^2 n}$ instead of $\poly(n)$.}
\comment{ If it were only for this part, we would just take $\gamma$ to be a constant, but the problem is that what we are interested in is.}

It now remains to bound $\pr{\eventfont{H}^c, S \leq s}$. Fix $x \in \{1, \dots, n\}$, we have
\begin{align*}
\pr{\sum_{k=1}^S \1(Y_k \leq x) > \gamma x/\mu, S \leq s}
&\leq \sum_{j=1}^s \pr{Y_j = x, \event{\forall i < j, Y_i > x}, j < S, \sum_{k=1}^S \1(Y_k \leq x) > \gamma x/\mu  } \\
&\leq \sum_{j=1}^s \pr{Y_j = x, \event{\forall i < j, Y_i > x}, j < S, \sum_{k=j+1}^{S_j} \1(Y_k \leq x) \geq \gamma x/\mu  } \\
&\leq \sum_{j=1}^s \pr{M \leq x} \cdot \pr{\sum_{k=j+1}^{S_j} \1(Y_k \leq x) \geq \gamma x/\mu  |Y_j = x, j < S},
\end{align*}
where we defined $S_j = \min\{ s \geq j + 1 : Y_k \geq n/2\}$.
To obtain the last inequality, we simply used the fact that $\event{Y_j = x, j < S} \subseteq \event{M \leq x}$. Moreover, $\event{j < S}$ can be determined by looking at $Y_1, \dots, Y_j$ and thus conditioned on $\event{Y_j=x}$, $Y_k$ for $k \geq j+1$ and also $S_j$ are independent of $\event{j < S}$. This means that we can drop $\event{j < S}$ from the conditioning.
 
To bound $\pr{M \leq x}$, we use \eqref{eq:bound-prob-m}. We can also bound $Y_k$ by a simpler random walk $Y'_k$ that moves forward with probability $3/4$, as we did in the proof of Lemma \ref{lem:bound-prob-s}. Thus, we obtain
\begin{align*}
\pr{\sum_{k=1}^S \1(Y_k \leq x) > \gamma x/\mu, S \leq s}
&\leq \frac{1}{2^{\ell} \binom{n}{\ell}} (2n)^x \cdot s \cdot \pr{\sum_{k=1}^{\infty} \1(Y'_k \leq x) \geq \gamma x/\mu | Y'_0 = 0} \\
&\leq \frac{1}{2^{\ell} \binom{n}{\ell}} (2n)^x \cdot s \cdot 2 \exp{-\frac{\mu (\gamma - 2) x}{2}},
\end{align*}
where we used \cite[Lemma A.5]{HL09}.
As a result, by a union bound,
\begin{align*}
\pr{\eventfont{H}^c, S \leq s} &\leq \frac{1}{2^{\ell} \binom{n}{\ell} } \cdot 2s \cdot \sum_{x=1}^n \exp{x\left(\log(2n) - \frac{\mu (\gamma - 2)}{2}\right)} \\
&\leq \frac{1}{2^{\ell} \binom{n}{\ell} } \cdot \frac{1}{\poly n},
\end{align*}
where to get the last inequality, we choose $\gamma = c' \log n$ for large enough $c'$ and use the fact that  $s$ will be chosen linear in $n$.
Continuing, we reach
\begin{align*}
\pr{W_1 + \dots W_S > t, S \leq s} 
&\leq \pr{W_1 + \dots W_S > t, S \leq s, \eventfont{H}} + \pr{\eventfont{H}^c, S \leq s} \\
&\leq \frac{1}{2^{\ell} \binom{n}{\ell}} \frac{1}{\poly(n)}.
\end{align*}
\end{proof}
To complete the proof of Lemma \ref{lem:time-reach-middle}, we just plug the bounds obtained from Lemma \ref{lem:bound-prob-s} with $s=16n$ and from Lemma \ref{lem:bound-waiting-time} into equation \eqref{eq:t-s}.
\end{proof}

\section{An additional lemma}
Consider a random walk on a line indexed from $-1$ to $a$. At positions $i > 0$, the probability of moving to the right is $p_{+}(i)$ (depending on $i$ and for points $i \leq 0$, the probability of moving to the right is $p_{-}$. The following lemma gives a bound on the probability of hitting the node $-1$ before hitting $a$ when starting at position $0$. In our setting, we are interested in the case where $p_{-}$ and $p_{+}$ are (significantly) larger than $1/2$ so that the probability of hitting $-1$ before $a$ is small.

\begin{lemma}
\label{lem:hitting-prob}
Assume $p_+(i), p_{-} > 1/2$. Then the probability of hitting $-1$ before $a$ is exactly
\[
\frac{1}{1+ \alpha_{-} \cdot \frac{\prod_{j=1}^{a-1} \alpha_+(j)}{1 + \sum_{i=1}^{a-1} \prod_{j=i}^{a-1} \alpha_+(j)} } \ ,
\] 
where $\alpha_+(i) = \frac{p_{+}(i)}{1-p_+(i)}$ and $\alpha_- = \frac{p_-}{1-p_-}$. In particular, if $\alpha_+(i) = \alpha_+$ for all $i$, this probability becomes
\[
\frac{1}{1+ \alpha_{-} \cdot \frac{\alpha_+^{a} - \alpha_+^{a-1}}{\alpha_+^{a} - 1}  } \leq \frac{1}{1+ \alpha_{-} \cdot (1-1/\alpha_+)  } \ .
\] 

%\[
%\frac{1}{1 + \frac{2p_+ - 1}{p_+} \frac{p_{-}}{2p_{-} - 1} \left( \left(\frac{p_-}{1-p_-}\right)^b - 1 \right)}.
%\] 
%\frac{p_+}{2p_+ - 1} \cdot \left(\frac{1-p_{-}}{p_{-}} \right)^b$
%$\frac{1}{\left(1 + \frac{\alpha_+ - 1}{\alpha_+}  \cdot (\alpha_-^b - 1) \right)}$.
\end{lemma}
\begin{proof}
Let $P_i$ be the probability of first reaching $-1$ when starting at position $i$. We can write for any %$i \in [-b+1, 0]$,
%$P_{i} = p_{-} P_{i+1} + (1-p_{-}) P_{i-1}$, which can be re-written as
%\[
%\frac{p_{-}}{1-p_{-}}  \left( P_{i} - P_{i+1} \right) = \left(P_{i-1} - P_i \right).
%\]
for $i \in [1, a-1]$, $P_{i} = p_{+}(i) P_{i+1} + (1-p_{+}(i)) P_{i-1}$, which can be re-written as
\[
\frac{p_{+}(i)}{1-p_{+}(i)} \left(P_{i} - P_{i+1} \right) =  \left(P_{i-1} - P_i \right).
\]
We now use the boundary condition at node $a$: $P_a = 0$. Thus, $ \left(P_{a-2} - P_{a-1} \right) =  \frac{p_{+}(a-1)}{1-p_{+}(a-1)} P_{a-1}$. Moreover, we see by induction that for any $i \geq 1$, $P_{i-1} - P_i = \left( \prod_{j=i}^{a-1}\frac{p_+(j)}{1-p_+(j)}\right) P_{a-1}$. We can now write a telescoping sum
\[
P_{0} - P_{a-1} = \sum_{i=1}^{a-1} P_{i-1} - P_{i} = \sum_{i=1}^{a-1} \prod_{j=i}^{a-1}\alpha_+(j) \cdot P_{a-1}.
\]
As a result, %if we write $\alpha_+ = \frac{p_+}{1-p_+}$,
\[
P_{0} = P_{a-1} \left( 1 + \sum_{i=1}^{a-1} \prod_{j=i}^{a-1}\alpha_+(j). \right).
\]
We can then write $P_{-1} - P_{0} = \frac{p_-}{1-p_-} \left( P_{0} - P_{1} \right) = P_{a-1} \cdot \prod_{j=1}^{a-1}\alpha_+(j) \cdot \frac{p_-}{1-p_-}$. 

%More generally, for any $i \in [-b+1, 0]$, we have 
%\[
%P_{i-1} - P_i = P_{a-1} \cdot \left(\frac{p_{+}}{1-p_{+}}\right)^{a-1} \cdot  \left( \frac{p_-}{1-p_-} \right)^{-i+1}.
%\]
%By writing a telescoping sum, we have
%\[
%P_{-b} - P_0 = \sum_{i=-b+1}^{0} P_{i-1} - P_i = P_{a-1} \left(\frac{p_{+}}{1-p_{+}}\right)^{a-1} \sum_{i=-b+1}^{0} \left( \frac{p_-}{1-p_-} \right)^{-i+1} = P_{a-1} \alpha_+^{a-1} \alpha_{-} \cdot \frac{\alpha^{b}_{-} - 1}{\alpha_{-} -1}.
%\]
Now, we use our second boundary condition $P_{-1} = 1$. We have
\begin{align*}
1 = P_{-1} &= P_0 +  P_{a-1} \cdot \alpha_{-} \prod_{j=1}^{a-1}\alpha_+(j) \\
		&= P_0 \left(1 + \alpha_{-} \frac{\prod_{j=1}^{a-1}\alpha_+(j)}{\sum_{i=1}^{a-1} \prod_{j=i}^{a-1}\alpha_+(j)} \right),		%&= P_0 \left(1 + \frac{\alpha^a_+ - \alpha^{a-1}_+}{\alpha_+^{a} - 1}  \cdot \frac{\alpha^{b + 1}_{-} - \alpha_-}{\alpha_{-} -1} \right).
\end{align*}
which leads to the desired result.
%This means that
%\begin{align*}
%P_0 &= \frac{1}{\left(1 + \frac{\alpha^a_+ - \alpha^{a-1}_+}{\alpha_+^{a} - 1}  \cdot \frac{\alpha^{b + 1}_{-} - \alpha_-}{\alpha_{-} -1} \right)} \\
%&\leq \frac{1}{\left(1 + \frac{\alpha_+ - 1}{\alpha_+}  \cdot \frac{\alpha_{-}}{\alpha_{-} - 1} \left( \alpha_-^b - 1 \right) \right)} \\ 
%&= \frac{1}{1 + \frac{(2p_+-1) p_{-}}{p_{+} (2 p_{-} - 1)} \left(\left( \frac{p_-}{1-p_-} \right)^b - 1\right)}.
%\end{align*}
%This means that for constant $\alpha_+ > 1$, $P_0 = O(\alpha_-^{-b}) = O\left( \left(\frac{1-p_{-}}{p_-} \right)^b \right)$.
\end{proof}

\bibliographystyle{alphaplus}
\bibliography{scrambling}

\end{document}